\documentclass[12pt, draftclsnofoot, onecolumn]{IEEEtran}

\usepackage{mathrsfs}
\usepackage{amsmath}
\usepackage{amsfonts}
\usepackage{amsthm}
\usepackage{amssymb}
\usepackage{bbm}
\usepackage{graphicx}
\usepackage{indentfirst}
\usepackage{array}
\usepackage{cite}
\usepackage{enumerate}
\usepackage{bm}

\usepackage{multirow}
\usepackage{multicol}
\usepackage{verbatim}
\usepackage{stfloats}
\usepackage{epstopdf}
\usepackage[ruled,vlined,linesnumbered,noresetcount]{algorithm2e}

\theoremstyle{definition}
\newtheorem{theorem}{\bf{Theorem}}
\newtheorem{lemma}{\bf{Lemma}}
\newtheorem{proposition}{\bf{Proposition}}

\newtheorem{definition}{\bf{Definition}}
\newtheorem{remark}{\bf{Remark}}

\begin{document}

\title{Performance and Construction of Polar Codes: The Perspective of Bit Error Probability}
\author{Bolin~Wu,~\IEEEmembership{Student Member,~IEEE,}
        Kai~Niu,~\IEEEmembership{Member,~IEEE,}
        Jincheng~Dai,~\IEEEmembership{Member,~IEEE}
        \vspace{-1.0em}

}

\maketitle

\begin{abstract}
Most existing works of polar codes focus on the analysis of block error probability. However, in many scenarios, bit error probability is also important for evaluating the performance of channel codes. In this paper, we establish a new framework to analyze the bit error probability of polar codes. Specifically, by revisiting the error event of bit-channel, we first introduce the conditional bit error probability as a metric to evaluate the reliability of bit-channel for both systematic and non-systematic polar codes. Guided by the concept of polar subcode, we then derive an upper bound on the conditional bit error probability of each bit-channel, and accordingly, an upper bound on the bit error probability of polar codes. Based on these, two types of construction metrics aiming at minimizing the bit error probability of polar codes are proposed, which are of linear computational complexity and explicit forms. Simulation results show that the polar codes constructed by the proposed methods can outperform those constructed by the conventional methods.
\end{abstract}


\IEEEpeerreviewmaketitle


\section{Introduction}

\IEEEPARstart{P}{olar} codes are the first practical codes that provably achieving the symmetric capacity of binary-input discrete memoryless channels (B-DMCs) asymptotically in the codelength under successive cancellation (SC) decoding \cite{arikan2009channel}. For the finite codelength, enhanced decoding algorithms such as successive cancellation list (SCL) decoding \cite{tal2015list,chen2012list} make polar codes competitive with LDPC and Turbo codes in addition to having low complexity. With the above aspects, polar codes have attracted enormous research interest with respect to both theory and practice.

The idea of polar codes is based on the concept of channel polarization, which transform independent copies of the transmission channel into a set of reliable and unreliable bit-channels. As one of the main focuses of polar coding, polar code construction boils down to efficiently identifying the reliability of each bit-channel and selecting the most reliable ones to transmit information bits, while the remained are fed with foreknown bits.

The existing polar construction methods fall into three main categories. In the first category, the channel conditions are used to evaluate the reliability of bit-channels. Ar{\i}kan \cite{arikan2009channel} first proposed a recursive algorithm based on the Bhattacharyya parameter, which is optimal only for binary erasure channels (BECs). For other B-DMCs, such as the binary symmetric channel (BSC) and binary-input additive white Gaussian noise (BI-AWGN) channel, the complexity of this method grows exponentially in codelength. In \cite{mori2009performance}, Mori and Tanaka proposed using density evolution (DE) to track the intermediate likelihood ratio and calculate the error probability of bit-channel. Though with theoretical guarantees on accuracy, this method poses high memory usages. Later, Tal and Vardy \cite{tal2013construct} used quantization operations to obtain the lower and upper bounds of the error probability of bit-channel, which can achieve good accuracy with sufficient quantization levels. Alternatively, Trifonov \cite{trifonov2012efficient} applied the Gaussian approximation (GA) of DE to further reduce the computational complexity. The aforementioned methods require the channel conditions to be known in advance. Such non-universality together with the high-dimensional calculation from channel stage to decision stage will pose huge challenges in practical application when a large range of code and channel parameters is envisaged.

As a practical alternative, another category of construction methods that utilize the deterministic properties of polar codes was derived. For instance, the polarized weight (PW) \cite{he2017beta} based on the partial order \cite{schurch2016partial} was proposed to sort the reliability of bit-channels by a function of their indices. Following this, the 5G standard uses a construction look-up table named Polar sequence \cite{3gpp.38.212} which sorts 1024 bit-channels in ascending order of reliability. Due to the independence of channel conditions, this kind of construction methods can greatly reduce the computational load but comes at the expense of performance loss.

The third category is based on the weight distributions of polar codes. In \cite{valipour2013probabilistic}, the probabilistic weight distributions of polar codes were proposed, which can be adopted to optimize the bit-channel selection for short polar codes. Apart from this, it was recently found that the error probability of bit-channel under SC decoding is tightly associated with a polar subcode \cite{niu2019polar}. Based on the weight distribution of polar subcode, a theoretical framework was established to analyze and construct polar codes. These methods exploit the algebraic features of polar codes and can be further extended to fading channels \cite{niu2020polar,niu2021polar}.

Most existing construction techniques focus on minimizing the block error probability of polar codes. However, bit error probability is also an important performance criterion in many scenarios such as data channels. And less research has been carried out on the performance analysis and construction of polar codes in terms of the bit error probability. In this paper, we move a step further by studying the bit error probability of polar codes with the concept of polar subcode. The main contributions of this paper can be summarized as follows.
\begin{itemize}
    \item \textbf{Establish a new framework to analyze the bit error probability of polar codes:} By revisiting the error event of bit-channel, we derive the conditional bit error probability of bit-channel to analyze the bit error probability of polar codes under SC decoding. Guided by the concept of polar subcode, the input-output weight enumerating function (IOWEF) for both systematic and non-systematic codings are introduced, with which the conditional bit error probability of bit-channel and the bit error probability of polar codes can be further bounded by the union bound. The proposed bounds have a closed-form and serve as a useful theoretical tool for accessing the bit error performance and gaining insight into the effect of constructive features of polar codes.
    \item \textbf{Design recursive algorithms to calculate the IOWEF of polar subcode:} For polar codes with small codelength, the IOWEF of polar subcode can be easily enumerated. While for the other cases, the enumeration can be a formidable problem. To tackle this issue, we exploit the Plotkin and symmetric structures of polar codes and then design recursive algorithms to calculate the IOWEF of polar subcode for both systematic and non-systematic polar codes.
    \item \textbf{Propose two types of polar code construction metrics:} Based on the conditional bit error probability upper bound of bit-channel, we propose two types of construction metrics named union-Bhattacharyya bound weight of the bit error probability (UBWB) and simplified UBWB (SUBWB), which aim at minimizing the bit error probability of polar codes. The UBWB and SUBWB have explicit forms and linear computational complexity since the IOWEF of polar subcode can be precomputed and stored. Simulation results show that the polar codes constructed by the proposed metrics can offer superior performance to those constructed by conventional methods.
\end{itemize}

The rest of the paper is organized as follows. In Section II, we present the preliminaries of polar codes, including polar coding, decoding, and polar subcode. Section III introduces the general framework for analyzing the bit error probability of polar codes. Following this, the detailed analysis of systematic and non-systematic polar codes are derived in Sections IV and V, respectively. Two types of construction metrics are introduced in Section VI. Section VII offers the numerical analysis and simulation results. Finally, we conclude our paper in Section VIII.

\emph{Notation Conventions:} In this paper, sets are denoted by the calligraphy letters, such as $\mathcal X$, and its cardinality is denoted as $\left| \mathcal X \right|$. Especially, we may also use hollow symbols (e.g., $\mathbb D$) to denote codeword sets. We use lowercase letters (e.g. $x$) to denote scalars. Notation $v_i^j$ denotes the vector $\left(v_i,v_{i+1},\cdots,v_{j-1},v_j\right)$ and $v_{\mathcal X}$ represents a vector with elements $v_i, i \in \mathcal X$. We may occasionally use the boldface lowercase letters (e.g. $\mathbf{v}$) to denote vectors. The boldface capital letters (e.g. $\mathbf{A}$) are used to denote matrices. The element in the $i$-th row and the $j$-th column of the matrix $\mathbf{A}$ is written as $\mathbf{A}_{i,j}$, and $\mathbf{A}_{\mathcal X,\mathcal Y}$ indicates the submatrix of $\mathbf{A}$ with rows from set $\mathcal X$ and columns from set $\mathcal Y$. The Hamming weight of a codeword $\mathbf{c}$ is denoted as $wt\left( \mathbf{c} \right)$, and $d_H\left( \mathbf{c}_1,\mathbf{c}_2 \right)$ means the Hamming distance between $\mathbf{c}_1$ and $\mathbf{c}_2$. Clearly, $wt\left( \mathbf{c} \right) =d_H\left( \mathbf{0},\mathbf{c} \right)$.

\section{Preliminaries}

\subsection{Polar Codes}
Given the codelength $N=2^n$, polar codes are performed by
\begin{equation}
    \label{x = uG}
    x_1^N = u_1^N{\bf{F}}_2^{ \otimes n}{{\bf{B}}_N},
\end{equation}
where $x_1^N$ is the codeword, $u_1^N$ is the source word, ${\mathbf{F}}_2^{ \otimes n}$ denotes the $n$-th Kronecker power of ${{\bf{F}}_2} = \left[ { \begin{smallmatrix} 1 & 0 \\ 1 &  1 \end{smallmatrix} } \right]$ and ${{\bf{B}}_N}$ is the bit-reversal operation matrix. Since ${{\mathbf{B}}_N}$ is only a simple permutation on $x_1^N$, we will use ${{\bf{G}}_N} = {\bf{F}}_2^{ \otimes n}$ as the generator matrix in this paper, which is also adopted in the 5G standard \cite{3gpp.38.212}. Let $W:{\mathcal X} \to {\mathcal Y}$ denote a B-DMC with input alphabet $\mathcal{X}$ and output alphabet $\mathcal{Y}$, the channel transition probabilities are given by $W\left( y|x \right)$. The well-known channel polarization phenomenon is based on channel combining and splitting, which transforms $N$ copies of $W$ into a synthetic channel ${W_N}$ by
\begin{equation}
    W_N\left( y_{1}^{N}|u_{1}^{N} \right) =W^N\left( y_{1}^{N}|u_{1}^{N}\mathbf{G}_N \right) =W^N\left( y_{1}^{N}|x_{1}^{N} \right) =\prod_{i=1}^N{W\left( y_i|x_i \right)}.
\end{equation}
The synthetic channel ${W_N}$ are splitted into $N$ bit-channels $W_N^{\left( i \right)}$ with transition probabilities
\begin{equation}
    W_{N}^{\left( i \right)}\left( y_{1}^{N},u_{1}^{i-1}|u_i \right) =\sum_{u_{i+1}^{N}}{\frac{1}{2^{N-1}}W_N\left( y_{1}^{N}|u_{1}^{N} \right)},
\end{equation}
and they show a polarization effect in reliability. Given the code dimension $K$, the information bits $\mathbf{b}$ are transmitted over the $K$ most reliable bit-channels, collectively referred to as the information set $\mathcal{A}$. Its complementary set $\mathcal {A}^{c}$ denotes the frozen set and the corresponding frozen bits are set to fixed values known by both encoder and decoder. Note that for symmetric channels, the performance of polar codes is independent of the choice of frozen bits \cite{arikan2009channel}. For convenience, we employ all-zero frozen bits in the following analysis.

Having received the channel output $y_1^N$, a successive cancellation (SC) decoding can be adopted to decode polar codes. The estimate $\hat u_1^N$ of $u_1^N$ is determined in serial order from index 1 to $N$. However, the performance under SC decoding is not competitive for finite code length. Several enhanced decoding algorithms are proposed to improve the coding performance, such as successive cancellation list (SCL) decoding, successive cancellation stack (SCS) decoding, etc.

For systematic polar codes, the information bits are encoded into a codeword in such a way that they appear transparently as part of codeword. As shown in \cite{arikan2011systematic}, the index set of information bits in a codeword $x_{1}^{N}$ can be chosen equal to the information set $\mathcal{A}$. Hence, a codeword can be split into two parts and rewritten as $x_{1}^{N}=\left( x_{\mathcal{A}},x_{\mathcal{A} ^c} \right)$, where ${x_{{\mathcal A}}}$ and ${x_{{{{\mathcal A}}^c}}}$ denote the information bits and parity bits, respectively. Ar{\i}kan shows that the systematic variant of polar coding preserves the BLER performance and meanwhile outperforms in terms of BER performance.

\subsection{Subcode and Polar subcode}
The block error probability under SC decoding is upper bounded by the sum of the error probabilities over the bit-channels $\left\{ {W_N^{\left( i \right)}} \right\}$ corresponding to the information set. In the previous work \cite{niu2019polar}, it was found that the error probability of bit-channel $W_N^{\left( i \right)}$ with codelength $N$ under SC decoding is tightly associated with a subcode $\mathbb C_N^{\left( i \right)}$ and the corresponding polar subcode $\mathbb D_N^{\left( i \right)}$, which are defined respectively as follows
\begin{definition}\label{def_WEF_polar_subcode}
Given the codelength $N$, the subcode $\mathbb C_N^{\left( i \right)}$ and polar subcode $\mathbb D_N^{\left( i \right)}$ of the $i$-th bit-channel are definied by
\begin{equation}\label{def_subcode}
    \mathbb C_N^{\left( i \right)} \triangleq \left\{ {{\bf{c}}:{\bf{c}} = \left( {0_1^{i - 1},u_i^N} \right){{\bf{G}}_N},\forall u_i^N \in {{\cal X}^{N - i + 1}}} \right\},
\end{equation}
\begin{equation}\label{def_polar_subcode}
    \mathbb D_N^{\left( i \right)} \triangleq \left\{ {{{{\bf{c}}^{\left( 1 \right)}}}:{{\bf{c}}^{\left( 1 \right)}} = \left( {0_1^{i - 1},1, u_{i+1}^N} \right){{\bf{G}}_N},\forall u_{i+1}^N \in {{\cal X}^{N - i}}} \right\}.
\end{equation}
\end{definition}

Obviously, $\mathbb C_N^{\left( i \right)}$ is an $\left( {N,N - i + 1} \right)$ linear block code and its (output) weight enumerating function (WEF) is denoted as
\begin{equation}
    S_N^{\left( i \right)}\left( Z \right) = \sum\limits_d {S_N^{\left( i \right)}\left( d \right){Z^d}},
\end{equation}
where $S_N^{\left( i \right)}\left( d \right)$ is the number of codewords with output weight $d$ in $\mathbb C_N^{\left( i \right)}$. Similarly, define the polar WEF of $\mathbb D_N^{\left( i \right)}$ (also called polar spectrum) by the following polynomial with $A_N^{\left( i \right)}\left( d \right)$ being the number of codewords with output weight $d$ in $\mathbb D_N^{\left( i \right)}$.
\begin{equation}\label{def_polar_WEF}
    A_N^{\left( i \right)}\left( Z \right) = \sum\limits_d {A_N^{\left( i \right)}\left( d \right){Z^d}},
\end{equation}

In what follows, we will use $\{ {S_N^{\left( i \right)}\left( d \right)} \}$ and $\{ {A_N^{\left( i \right)}\left( d \right)} \}$ as shorthands for WEF of $\mathbb C_N^{\left( i \right)}$ and polar spectrum of $\mathbb D_N^{\left( i \right)}$, respectively. For codes of small codelengths, $\{ {S_N^{\left( i \right)}\left( d \right)} \}$ and $\{ {A_N^{\left( i \right)}\left( d \right)} \}$ can be enumerated easily. While for the other cases, they can be calculated by a recursive algorithm proposed in \cite{niu2019polar}.

Let $\mathcal{E} _i=\left\{ \left( u_{1}^{N},y_{1}^{N} \right) :W_{N}^{\left( i \right)}\left( y_{1}^{N},u_{1}^{i-1}|u_i \right) \leqslant W_{N}^{\left( i \right)}\left( y_{1}^{N},u_{1}^{i-1}|u_i\oplus 1 \right) \right\}$ denote the error event of the $i$-th bit-channel \cite{arikan2009channel}. By invoking the union bound \cite{niu2019polar}, the error probability of $W_N^{\left( i \right)}$ can be upper bounded by
\begin{equation}\label{Error_Probability_BitChannel_d}
    {P_e}\left( {W_N^{\left( i \right)}} \right) = P_e\left( \mathcal{E} _i \right) \le \sum\limits_{d = 1}^N {A_N^{\left( i \right)}\left( d \right)P_N^{\left( i \right)}\left( d \right)},
\end{equation}
where $P_N^{\left( i \right)}\left( d \right)$ is the pairwise error probability between the all-zero codeword and the codeword with weight $d$ in $\mathbb D_N^{\left( i \right)}$.

Furthermore, let $\mathcal{E} =\left\{ \left( u_{1}^{N},y_{1}^{N} \right) \in \mathcal{X} ^N\times \mathcal{Y} ^N:u_{\mathcal{A}}\ne \hat{u}_{\mathcal{A}} \right\}$ denote the block error event under SC decoding. Since SC decoding performs in a serial manner from the first bit-channel to the last, given the information set $\mathcal{A}$, then $\mathcal{E} \subset \bigcup\nolimits_{i\in \mathcal{A}}^{}{\mathcal{E} _i}$ and the block error probability of polar codes is upper bounded by
\begin{equation} \label{BLER_union_bound_of_SC}
    P_e\left( N,K,\mathcal{A} \right) =P_e\left( \mathcal{E} \right) \le \sum\limits_{i \in {\cal A}} {{P_e}\left( {W_N^{\left( i \right)}} \right)}  \le \sum\limits_{i \in {\cal A}} {\sum\limits_{d = 1}^N {A_N^{\left( i \right)}\left( d \right)P_N^{\left( i \right)}\left( d \right)}}.
\end{equation}

The above bounds have explicit expressions which reflect the underlying relation between the error performance of polar codes and the algebraic features of polar subcodes.

\section{Conditional Bit Error Bit Probability}\label{BER_total}

The main concern of interest in the following analysis is how to evaluate the bit error probability of polar codes. To tackle this issue, we introduce the conditional bit error probability of bit-channel which is defined as follows.

\begin{definition}
Given the polar code parameter $\left( N,K,\mathcal{A} \right)$, the conditional bit error probability of the $i$-th bit-channel under SC decoding is defined by
\begin{equation}
   P_{b}^{\mathcal{A}}\left( W_{N}^{\left( i \right)} \right) =\frac{1}{K}E\left[ d_H(b_{1}^{K},\hat{b}_{1}^{K})|\mathcal{E} _i \right] P_e\left( W_{N}^{\left( i \right)} \right),
\end{equation}
where $b_{1}^{K}$ is the information bits and $\hat{b}_{1}^{K}$ is its estimate, $E\left[ d_H(b_{1}^{K},\hat{b}_{1}^{K})|\mathcal{E} _i \right]$ is the conditional expectation of $d_H( b_{1}^{K},\hat{b}_{1}^{K} )$ given the event $\mathcal{E} _i$.
\end{definition}


\begin{theorem}\label{BER_SC_polar}
For a polar code with parameter $\left( N,K,\mathcal{A} \right)$, the bit error probability of the polar code under SC decoding is upper bounded by
\begin{equation}
    P_b\left( N,K,\mathcal{A} \right) \leqslant \sum_{i\in \mathcal{A}}{P_{b}^{\mathcal{A}}\left( W_{N}^{\left( i \right)} \right)}.
\end{equation}
\end{theorem}
\begin{proof}
Based on the successive manner of SC decoding, any $u_i\ne \hat{u}_i$ will commit an erroneous estimate $\hat{b}_{1}^{K}$. Hence $u_{\mathcal{A}}\ne \hat{u}_{\mathcal{A}}$ is equivalent to $x_{\mathcal{A}}\ne \hat{x}_{\mathcal{A}}$ and this also validates the reason why the block error probability of systematic and non-systematic polar codes are the same. According to \cite[Sec. V]{arikan2009channel}, the block error event $\mathcal{E}$ can be expressed as $\mathcal{E} =\bigcup\nolimits_{i\in \mathcal{A}}^{}{\mathcal{B} _i}$, where $\mathcal{B} _i\triangleq \left\{ \left( u_{1}^{N},y_{1}^{N} \right) \in \mathcal{X} ^N\times \mathcal{Y} ^N:u_{1}^{i-1}=\hat{u}_{1}^{i-1},u_i=h_i\left( y_{1}^{N},\hat{u}_{1}^{i-1} \right) \right\}$ is the event that the first decision error under SC decoding occurs at the $i$-th bit-channel with $h_i\left( y_{1}^{N},\hat{u}_{1}^{i-1} \right)$ denotes the decision function. Thus, we have
\begin{equation}\label{BLER_equality}
    P_e\left( N,K,\mathcal{A} \right) =P_e\left( \mathcal{E} \right) =\sum_{i\in \mathcal{A}}{P_e\left( \mathcal{B} _i \right)}.
\end{equation}

Note that for each $\mathcal{B} _i$, the corresponding bit error probability is $\frac{1}{K}E\left[ d_H( b_{1}^{K},\hat{b}_{1}^{K} ) |\mathcal{B} _i \right] P_e\left( \mathcal{B} _i \right) $, and this together with (\ref{BLER_equality}) implies that the bit error probability of polar codes with parameter $\left( N,K,\mathcal{A} \right)$ under SC decoding is
\begin{equation}\label{BER_equality}
P_b\left( N,K,\mathcal{A} \right) =\sum_{i\in \mathcal{A}}{\frac{1}{K}E\left[ d_H( b_{1}^{K},\hat{b}_{1}^{K} ) |\mathcal{B} _i \right] P_e\left( \mathcal{B} _i \right)}.
\end{equation}
Since $\mathcal{B} _i \subset \mathcal{E} _i$, and
\begin{equation}
    E\left[ d_H(b_{1}^{K},\hat{b}_{1}^{K})|\mathcal{B} _i \right] P_e\left( \mathcal{B} _i \right) =\hspace{-0.5em}\sum_{d_H(b_{1}^{K},\hat{b}_{1}^{K})}{\hspace{-0.4em}}d_H(b_{1}^{K},\hat{b}_{1}^{K})P\left( \mathcal{B} _i|d_H(b_{1}^{K},\hat{b}_{1}^{K}) \right) P\left( d_H(b_{1}^{K},\hat{b}_{1}^{K}) \right),
\end{equation}
it follows that
\begin{equation}\label{Bi_Ei_inequality}
E\left[ d_H(b_{1}^{K},\hat{b}_{1}^{K})|\mathcal{B} _i \right] P_e\left( \mathcal{B} _i \right) \leqslant E\left[ d_H(b_{1}^{K},\hat{b}_{1}^{K})|\mathcal{E} _i \right] P_e\left( W_{N}^{\left( i \right)} \right).
\end{equation}
Substituting (\ref{Bi_Ei_inequality}) into (\ref{BER_equality}), we obtain
\begin{equation}
    P_b\left( N,K,\mathcal{A} \right) \leqslant \sum_{i\in \mathcal{A}}{\frac{1}{K}E\left[ d_H(b_{1}^{K},\hat{b}_{1}^{K})|\mathcal{E} _i \right] P_e\left( W_{N}^{\left( i \right)} \right)}=\sum_{i\in \mathcal{A}}{P_{b}^{\mathcal{A}}\left( W_{N}^{\left( i \right)} \right)}.
\end{equation}
This completes the proof.
\end{proof}

An intuitive explanation of Theorem \ref{BER_SC_polar} is that the final decision $\hat{b}_{1}^{K}$ is error-free only if all the bit-channels with indices in $\mathcal{A}$ are decoded correctly. Together with the bit error probability caused by the decoding error of each bit-channel, the probability of its complementary event thus gives an upper bound on the bit error probability of polar codes under SC decoding, which is denoted by the sum of conditional bit error probability of bit-channels. In the following sections, we will give a detailed analysis for both systematic and non-systematic polar codes, and investigate the relations between conditional bit error probability and polar subcode.

\section{Bit Error Probability Analysis of Systematic Polar Codes}\label{BER_SPC}

In this section, we focus on the analysis of systematic polar codes. First, we introduce the IOWEFs of the subcode and polar subcode for systematic coding. Second, we investigate the homogeneous property of polar subcode. Based on these, we then derive an upper bound on the conditional bit error probability of systematic polar codes. Furthermore, a recursive algorithm is proposed to calculate the IOWEF of polar subcode.

\subsection{Conditional Bit Error Probability of Bit-channel for Systematic Polar Codes}

Exploiting the concept of polar subcode and its WEF, the error probability of bit-channel can be expressed in the form of the union bound as \eqref{Error_Probability_BitChannel_d}. Recall that the bit error probability is another commonly used performance measure of which the upper bound is expressible in terms of the IOWEFs of the codes \cite{ryan2009channel}. Regarding this, we introduce the following definition.

\begin{definition}
The IOWEF and polar IOWEF of the $i$-th bit-channel are defined respectively by
\begin{equation}
    S_{N}^{\left( i \right)}\left( W,Z \right) =\sum_w{\sum_d{S_{N}^{\left( i \right)}\left( w,d \right) W^wZ^d}},
\end{equation}
\begin{equation}\label{def_polar_IOWEF_SPC}
    A_{N}^{\left( i \right)}\left( W,Z \right) =\sum_w{\sum_d{A_{N}^{\left( i \right)}\left( w,d \right) W^wZ^d}},
\end{equation}
where $S_N^{\left( i \right)}\left( {w,d} \right)$ denotes the number of codewords with information (input) weight $w$ and codeword (output) weight $d$ in $\mathbb C_N^{\left( i \right)}$, and $A_N^{\left( i \right)}\left( {w,d} \right)$ holds a similar definition for $\mathbb D_N^{\left( i \right)}$.
\end{definition}

 Regarding the definition in \eqref{def_subcode}, the subcode $\mathbb C_N^{\left( i \right)}$ is actually an $\left( {N,N - i + 1} \right)$ polar code with information set $\mathcal{A} = \{ i,i + 1, \cdots ,N\}$. Hence for systematic coding, we denote the input weight of codeword $\bf{c}$ in $\mathbb C_N^{\left( i \right)}$ as the Hamming weight of the last $N-i+1$ codeword bits, i.e., $w = {wt}\left( {{\bf{c}}_i^N} \right)$. And so is the definition for $\mathbb D_N^{\left( i \right)}$, since the polar subcode is a subset of subcode. Convenient shorthands for IOWEF of $\mathbb{C} _{N}^{\left( i \right)}$ and polar IOWEF of $\mathbb{D} _{N}^{\left( i \right)}$ are $\left\{ S_{N}^{\left( i \right)}\left( w,d \right) \right\}$ and $\left\{ A_{N}^{\left( i \right)}\left( w,d \right) \right\}$, respectively. The subcode and polar subcode have the following properties.

\begin{proposition}\label{subcode_duality}
The subcode $\mathbb C_N^{\left( {N + 2 - i} \right)}$ is the dual code of $\mathbb C_N^{\left( {i} \right)}$, where $N/2 + 1 \le i \le N$. In other words, $\mathbb C_N^{\left( {N + 2 - i} \right)} = \mathbb C_N^{\left( i \right) \bot }$. Especially, $\mathbb C_N^{\left( {N/2 + 1} \right)}$ is a self-dual code.
\end{proposition}

\begin{proposition}\label{cyclic_code}
The subcode $\mathbb C_N^{\left( {i} \right)}$ and polar subcode $\mathbb D_N^{\left( {i} \right)}$ are both cyclic codes, where $1 \le i \le N$.
\end{proposition}

The proofs of Proposition \ref{subcode_duality} and Proposition \ref{cyclic_code} are given in \cite{niu2019polar} and Appendix \ref{proof_of_cyclic_code}, respectively. Let $\mathbf{M}$ be the ${2^{N - i}} \times N$ matrix whose rows are all codewords in $\mathbb D_N^{\left( i \right)}$, and let $\mathbf{M}_d$ be the $A_{N}^{\left( i \right)}\left( d \right) \times N$ submatrix of $\mathbf{M}$ consisting of the codewords of weight $d$.

\begin{proposition}\label{homogeneous}
The polar subcode $\mathbb D_N^{\left( i \right)}$ is homogeneous \cite{huffman2010fundamentals}, which means for $0 \le d \le N$, each column of $\mathbf{M}_d$ has the same Hamming weight.
\end{proposition}
\begin{proof}
Consider a codeword with Hamming weight $d$ in $\mathbb D_N^{\left( i \right)}$, let ${{\mathbf{M}'}_d}$ be the $m'\times N$ matrix consisting of it and the other $\left( m'-1 \right)$ codewords obtained by its cyclic shift. Apparently, ${{\mathbf{M}'}_d}$ is a submatrix of $\mathbf{M}_d$ with the Hamming weight of each column equal to $\frac{d}{N}m'$. By considering all the codewords with output weight $d$, Proposition \ref{homogeneous} follows immediately with each column's Hamming weight of $\mathbf{M}_d$ equal to $\frac{d}{N}A_{N}^{\left( i \right)}\left( d \right)$. In addition, one can also prove with a similar approach that the subcode $\mathbb C_N^{\left( i \right)}$ is also homogeneous.
\end{proof}

\begin{remark}
Let $\mathbf{M}_{d}^{*}$ be the matrix obtained from $\mathbf{M}_{d}$ by selecting the column with index in a set $\mathcal{Q} \subset \left\{ 1,2,\ldots ,N \right\}$. The dimension of $\mathbf{M}_{d}^{*}$ is hence $A_{N}^{\left( i \right)}\left( d \right) \times \left| \mathcal{Q} \right|$ and the number of nonzero entries is $\sum\nolimits_j^{}{w_j}$, where $w_j$ is the Hamming weight of the $j$-th row. Since $\mathbb D_N^{\left( i \right)}$ is homogeneous, we have $\sum\nolimits_j^{}{\frac{w_j}{\left| \mathcal{Q} \right|}}=\frac{d}{N}A_{N}^{\left( i \right)}\left( d \right)$ that both denote the Hamming weight of the column in $\mathbf{M}_{d}$ or $\mathbf{M}_{d}^{*}$.
\end{remark}

For systematic polar codes, the information bits are explicitly invisible in the codeword. Given the parameter $\left( N,K,\mathcal{A} \right)$, the conditional bit error probability is hence denoted by
\begin{equation}\label{SPC_BER_bit-channel}
    P_{b,sys}^{\mathcal{A}}\left( W_{N}^{\left( i \right)} \right) =\frac{1}{K}E\left[ d_H(\mathbf{x}_{\mathcal{A}},\mathbf{\hat{x}}_{\mathcal{A}})|\mathcal{E} _i \right] P_e\left( W_{N}^{\left( i \right)} \right),
\end{equation}
where $\hat{\mathbf{x}}_{\mathcal{A}}$ is the wrong estimate of information bits $\mathbf{x}_{\mathcal{A}}$. Based on the polar subcode and its homogeneous property, we can further bound $P_{b,sys}^{\mathcal{A}}\left( W_{N}^{\left( i \right)} \right)$ by the following Proposition.

\begin{proposition}\label{Prop_BER_Bound_BitChannel}
For symmetric channels, the conditional bit error probability ${P_{b,sys}\left( W_{N}^{\left( i \right)} \right)}$ of systematic polar codes can be upper bounded as follows
\begin{align}
    P_{b,sys}\left( W_{N}^{\left( i \right)} \right)
    & \leqslant \sum_d{\sum_w{\frac{w}{N-i+1}A_{N}^{\left( i \right)}\left( w,d \right) P_{N}^{\left( i \right)}\left( d \right)}} \label{Bit_Error_Probability_BitChannel_d_form1}\\
    & =\sum_d{\frac{d}{N}A_{N}^{\left( i \right)}\left( d \right) P_{N}^{\left( i \right)}\left( d \right)} \label{Bit_Error_Probability_BitChannel_d_form2}
\end{align}
\end{proposition}
\begin{proof}
By \cite[Prop. 13]{arikan2009channel}, if a B-DMC $W$ is symmetric, then the bit-channel $W_{N}^{\left( i \right)}$ is also symmetric. Hence, without loss of generality, we can assume $u_{1}^{N}=0_{1}^{N}$, such that the codeword $\mathbf{x}=u_{1}^{N}\mathbf{G}_N=0_{1}^{N}$ is also an all-zero bit vector. Then, by \eqref{SPC_BER_bit-channel}, we have
\begin{equation}
    \begin{aligned}
    P_{b,sys}^{\mathcal{A}}\left( W_{N}^{\left( i \right)} \right) & =\frac{1}{K}E\left[d_H(0_{1}^{K},\mathbf{\hat{x}}_{\mathcal{A}})|\left\{ \mathcal{E} _i,u_{1}^{N}=0_{1}^{N} \right\} \right] P\left( \mathcal{E} _i|u_{1}^{N}=0_{1}^{N} \right) \\
    & =\frac{1}{K}E\left[ wt( \mathbf{\hat{x}}_{\mathcal{A}} ) |\left\{ \mathcal{E} _i,u_{1}^{N}=0_{1}^{N} \right\} \right] P\left( \mathcal{E} _i|u_{1}^{N}=0_{1}^{N} \right) \\
    & =\frac{1}{K}\sum_{wt\left( \mathbf{\hat{x}}_{\mathcal{A}} \right)}{wt\left( \mathbf{\hat{x}}_{\mathcal{A}} \right) P\left( \left\{ \mathcal{E} _i,u_{1}^{N}=0_{1}^{N} \right\} |wt\left( \mathbf{\hat{x}}_{\mathcal{A}} \right) \right) \frac{P\left( wt\left( \mathbf{\hat{x}}_{\mathcal{A}} \right) \right)}{P\left( u_{1}^{N}=0_{1}^{N} \right)}}\label{deriv_BER_i_tmp1}
    \end{aligned}
\end{equation}

Let $\mathbb{C} _{N}^{\left( i \right)}-\mathbb{D} _{N}^{\left( i \right)}=\left\{ \mathbf{c}^{\left( 0 \right)}:\mathbf{c}^{\left( 0 \right)}=\left( 0_{1}^{i-1},0,u_{i+1}^{N} \right) \mathbf{G}_N,\forall u_{i+1}^{N}\in \mathcal{X}^{N-i} \right\}$. Then, for $\left\{ \mathcal{E} _i,u_{1}^{N}=0_{1}^{N} \right\}$, we have
\begin{equation}
    \begin{aligned}
    \left\{ \mathcal{E} _i,u_{1}^{N}=0_{1}^{N} \right\}
    &=\left\{ \left( 0_{1}^{N},y_{1}^{N} \right) :W_{N}^{\left( i \right)}\left( y_{1}^{N},0_{1}^{i-1}|0 \right) \leqslant W_{N}^{\left( i \right)}\left( y_{1}^{N},0_{1}^{i-1}|1 \right) \right\} \\
    &=\left\{ \left( 0_{1}^{N},y_{1}^{N} \right) :\sum_{\mathbf{c}^{\left( 0 \right)}}{W^N\left( y_{1}^{N}|\mathbf{c}^{\left( 0 \right)} \right)}\leqslant \sum_{\mathbf{c}^{\left( 1 \right)}}{W^N\left( y_{1}^{N}|\mathbf{c}^{\left( 1 \right)} \right)} \right\}.
    \end{aligned}
\end{equation}

As shown in \cite{miloslavskaya2014sequential,trifonov2021}, the above probabilities can be well approximated by $\sum\nolimits_{\mathbf{c}}^{}{W^N\left( y_{1}^{N}|\mathbf{c} \right)}\approx \max _{\mathbf{c}}W^N\left( y_{1}^{N}|\mathbf{c} \right) $, such that
\begin{equation}
    \begin{aligned}
    \left\{ \mathcal{E} _i,u_{1}^{N}=0_{1}^{N} \right\}
    & \approx \left\{ \left( 0_{1}^{N},y_{1}^{N} \right) :\underset{\mathbf{c}^{\left( 0 \right)}}{\max}\,\,W^N\left( y_{1}^{N}|\mathbf{c}^{\left( 0 \right)} \right) \leqslant \underset{\mathbf{c}^{\left( 1 \right)}}{\max}\,\,W^N\left( y_{1}^{N}|\mathbf{c}^{\left( 1 \right)} \right) \right\} \\
    & \subset \left\{ \left( 0_{1}^{N},y_{1}^{N} \right) :W^N\left( y_{1}^{N}|0_{1}^{N} \right) \leqslant \underset{\mathbf{c}^{\left( 1 \right)}}{\max}\,\,W^N\left( y_{1}^{N}|\mathbf{c}^{\left( 1 \right)} \right) \right\} \\
    & \subset \bigcup_{\mathbf{c}^{\left( 1 \right)}}{\left\{ \left( 0_{1}^{N},y_{1}^{N} \right) :W^N\left( y_{1}^{N}|0_{1}^{N} \right) \leqslant W^N\left( y_{1}^{N}|\mathbf{c}^{\left( 1 \right)} \right) \right\}}.
    \end{aligned}
\end{equation}

Let $\left\{ \tilde{\mathcal{E}}_{i,\mathbf{c}^{\left( 1 \right)}},u_{1}^{N}=0_{1}^{N} \right\} =\left\{ \left( 0_{1}^{N},y_{1}^{N} \right) :W^N\left( y_{1}^{N}|0_{1}^{N} \right) \leqslant W^N\left( y_{1}^{N}|\mathbf{c}^{\left( 1 \right)} \right) \right\}$ denote the event that all-zero codeword $0_{1}^{N}$ is transmitted and the decoder chooses $\mathbf{c}^{\left( 1 \right)}$. Then, $\left\{ \mathcal{E} _i,u_{1}^{N}=0_{1}^{N} \right\} \subset \bigcup_{\mathbf{c}^{\left( 1 \right)}}{\left\{ \tilde{\mathcal{E}}_{i,\mathbf{c}^{\left( 1 \right)}},u_{1}^{N}=0_{1}^{N} \right\}}$. Note that this event is derived under the assumption that only two codewords are involved: $\mathbf{c}^{\left( 1 \right)}$ and $0_{1}^{N}$ \cite{ryan2009channel}. And the conditional probability $P\left( \tilde{\mathcal{E}}_{i,\mathbf{c}^{\left( 1 \right)}}|u_{1}^{N}=0_{1}^{N} \right) =P\left( \tilde{\mathcal{E}}_{i,\mathbf{c}^{\left( 1 \right)}},u_{1}^{N}=0_{1}^{N} \right) /P\left( u_{1}^{N}=0_{1}^{N} \right)$ is the pairwise error probability, which denotes the probability that the decoder chooses $\mathbf{c}^{\left( 1 \right)}$ given that $0_{1}^{N}$ was transmitted. Since $P\left( \tilde{\mathcal{E}}_{i,\mathbf{c}^{\left( 1 \right)}}|u_{1}^{N}=0_{1}^{N} \right)$ depends only on the Hamming distance $d$ between $\mathbf{c}^{\left( 1 \right)}$ and $0_{1}^{N}$, we can abbreviate it as $P_{N}^{\left( i \right)}\left( d \right)$.

In view of these, we may rewrite \eqref{deriv_BER_i_tmp1} as
\begin{equation} \label{derivation_SPC_bit-channel}
    \begin{aligned}
    P_{b,sys}^{\mathcal{A}}\left( W_{N}^{\left( i \right)} \right)
    & \overset{\left( \mathrm{a} \right)}{\leqslant} \frac{1}{K}\hspace{-0.1em}\sum_{\mathbf{c}^{\left( 1 \right)}}\hspace{-0.2em}{\sum_{wt(\mathbf{\hat{x}}_{\mathcal{A}})}\hspace{-0.4em}{wt(\mathbf{\hat{x}}_{\mathcal{A}})P\left( \left\{ \tilde{\mathcal{E}}_{i,\mathbf{c}^{\left( 1 \right)}},u_{1}^{N}=0_{1}^{N} \right\} |wt(\mathbf{\hat{x}}_{\mathcal{A}}) \right) \frac{P\left( wt(\mathbf{\hat{x}}_{\mathcal{A}}) \right)}{P\left( u_{1}^{N}=0_{1}^{N} \right)}}} \\
    & \overset{\left( \mathrm{b} \right)}{=}\frac{1}{K}\hspace{-0.1em}\sum_{\mathbf{c}^{\left( 1 \right)}}\hspace{-0.2em}{\sum_{wt(\mathbf{\hat{x}}_{\mathcal{A}})}\hspace{-0.4em}{wt(\mathbf{\hat{x}}_{\mathcal{A}})P\left( wt(\mathbf{\hat{x}}_{\mathcal{A}})|\left\{ \tilde{\mathcal{E}}_{i,\mathbf{c}^{\left( 1 \right)}},u_{1}^{N}=0_{1}^{N} \right\} \right) P\left( \tilde{\mathcal{E}}_{i,\mathbf{c}^{\left( 1 \right)}}|u_{1}^{N}=0_{1}^{N} \right)}} \\
    & \overset{\left( \mathrm{c} \right)}{=}\sum_{\mathbf{c}^{\left( 1 \right)}}{\frac{wt(\mathbf{c}_{\mathcal{A}}^{\left( 1 \right)})}{K}}P\left( \tilde{\mathcal{E}}_{i,\mathbf{c}^{\left( 1 \right)}}|u_{1}^{N}=0_{1}^{N} \right) \\
    & \overset{\left( \mathrm{d} \right)}{=}\sum_d{\sum_w{\frac{w}{N-i+1}A_{N}^{\left( i \right)}\left( w,d \right) P_{N}^{\left( i \right)}\left( d \right)}} \\
    & \overset{\left( \mathrm{e} \right)}{=}\sum_d{\frac{d}{N}A_{N}^{\left( i \right)}\left( d \right) P_{N}^{\left( i \right)}\left( d \right)},
    \end{aligned}
\end{equation}
where inequalities in (a) follows from the union bound. Given any event $\left\{ \tilde{\mathcal{E}}_{i,\mathbf{c}^{\left( 1 \right)}},u_{1}^{N}=0_{1}^{N} \right\}$, the Hamming weight $wt( \hat{\mathbf{x}}_{\mathcal{A}} )$ is determined and equals to $wt( \mathbf{c}_{\mathcal{A}}^{\left( 1 \right)} )$, hence we have (c) from (b). Equalities in (d) and (e) follow from the homogeneous property of polar subcode that $\sum_w{\frac{w}{N-i+1}A_{N}^{\left( i \right)}\left( w,d \right)}=\frac{d}{N}A_{N}^{\left( i \right)}\left( d \right)$ both denote the Hamming weight of the column in $\mathbf{M}_d$. Since (d) and (e) are independent of $\mathcal{A}$, we abbreviate $P_{b,sys}^{\mathcal{A}}\left( W_{N}^{\left( i \right)} \right)$ as $P_{b,sys}\left( W_{N}^{\left( i \right)} \right)$.
\end{proof}

Proposition \ref{Prop_BER_Bound_BitChannel} indicates that the conditional bit error probability of bit-channel $P_{sys,b}\left( W_{N}^{\left( i \right)} \right)$ is still associated with the polar subcode $\mathbb{D} _{N}^{\left( i \right)}$. By \eqref{Bit_Error_Probability_BitChannel_d_form2}, it would suffice to calculate $P_{sys,b}\left( W_{N}^{\left( i \right)} \right)$ with the polar spectrum $\{ {A_N^{\left( i \right)}\left( d \right)} \}$ alone. However, \eqref{Bit_Error_Probability_BitChannel_d_form1} shows the underlying connections that exists between the bit error probability under SC decoding and the polar IOWEF of polar subcode. Meanwhile, the IOWEF and polar IOWEF of each bit-channel would be necessary for analyzing the weight distributions of concatenated polar coding schemes, like CRC-concatenated polar codes. This will  be investigated in the future work.

\subsection{Calculation of IOWEF and Polar IOWEF of Bit-channel for Systematic Polar Codes}

From the definitions of the subcode and polar subcode, we have $\mathbb{C} _{N}^{\left( i \right)}=\mathbb{D} _{N}^{\left( i \right)}\cup \mathbb{C} _{N}^{\left( i+1 \right)}$, but $S_{N}^{\left( i \right)}\left( w,d \right) \ne A_{N}^{\left( i \right)}\left( w,d \right) +S_{N}^{\left( i+1 \right)}\left( w,d \right)$. For systematic polar codes, to obtain the IOWEFs and polar IOWEFs of bit-channels, we introduce two aided weight distributions defined as follows.

\begin{definition}
Let $U_{N}^{\left( i \right)}\left( k,l,m,n \right)$ be the number of codewords with Hamming weight constraint $\left( k,l,m,n \right)$ in $\mathbb{C} _{N}^{\left( i \right)}$, collectively referred to as the 4-split-spectrum $\left\{ U_{N}^{\left( i \right)}\left( k,l,m,n \right) \right\}$, where $1\leqslant i\leqslant N$, and $k=0,1$, $l=0,1,\ldots ,N-i$, $m=0,1$, $n=0,1,\ldots ,i-2$ denote respectively the Hamming weight of the last information bit, the other information bits, the last parity bit and the other parity bits.
\end{definition}

\begin{definition}
Let $V_{N}^{\left( i+1 \right) \rightarrow \left( i \right)}\left( p,q \right)$ be the number of codewords with Hamming weight constraint $\left( p,q \right)$ in $\mathbb{C} _{N}^{\left( i \right)}$, collectively referred to as the 2-split-spectrum $\left\{ V_{N}^{\left( i+1 \right) \rightarrow \left( i \right)}\left( p,q \right) \right\}$, where $1\leqslant i\leqslant N-1$, $p=0,1,\ldots ,N-i+1$ is the sum of the Hamming weights of the last parity bit and all the information bits, and $q=0,1,\dots ,i-1$ is the Hamming weight of the other parity bits.
\end{definition}
\begin{figure}[htb]
  \centering{\includegraphics[scale=0.60]{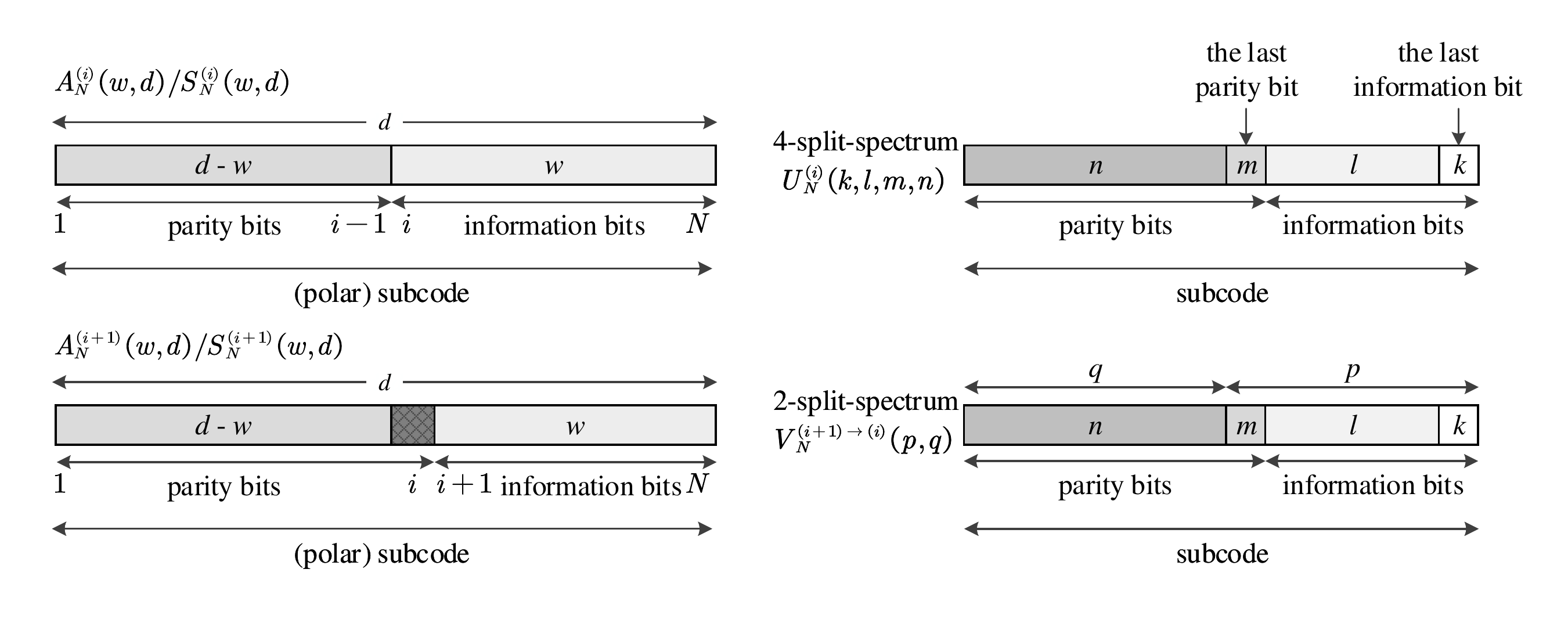}}
  \caption{Relations between the (polar) subcodes of two adjacent bit-channels.}\label{Definitions_of_split_spectrum}
\end{figure}

As shown in Fig. \ref{Definitions_of_split_spectrum}, the 4-split-spectrum is a subdivision of the IOWEF of subcode. From the above definitions, we further have
\begin{equation} \label{Rel_4_split_and_2_split}
    V_{2N}^{\left( j+1 \right) \rightarrow \left( j \right)}\left( p,q \right) =\sum_{k,l,m,n}{U_{2N}^{\left( j+1 \right)}\left( k,l,m,n \right) \mathbbm{1}_{\varOmega}},
\end{equation}
where $\varOmega =\left\{ \left( k,l,m,n \right) : k+l+m=p, n=q \right\}$.

\begin{theorem}
For the subcode $\mathbb{C} _{2N}^{\left( j \right)}$ and its dual code $\mathbb{C} _{2N}^{\left( j \right) \bot}=\mathbb{C} _{2N}^{\left( 2N+2-j \right)}$, where $2\leqslant j\leqslant N$, their 4-split spectra satisfy the following generalized MacWilliams identities
\begin{equation}\label{generalized_MacWilliams}
    U_{2N}^{\left( j \right)}\left( e,f,g,h \right) =\frac{1}{2^{j-1}}\sum_{e'=0}^1{\sum_{g'=0}^1{\sum_{f'=0}^{j-2}{\sum_{h'=0}^{2N-j}{U_{2N}^{\left( 2N+2-j \right)}\left( e',f',g',h' \right) \mathbf{K}_{eg'}^{\left( 1 \right)}}}}\mathbf{K}_{ge'}^{\left( 1 \right)}\mathbf{K}_{fh'}^{\left( 2N-j \right)}\mathbf{K}_{hf'}^{\left( j-2 \right)}},
\end{equation}
where $\mathbf{K}^{\left( M \right)}$ is the $M$-order Krawtchouk matrix with dimension $\left( M+1 \right) \times \left( M+1 \right)$, the entries are given by $\mathbf K_{xy}^{\left(M\right)}=\sum\limits_{z=0}^x (-1)^z \binom{y}{z} \binom{M-y}{x-z}$ with the indices $x$ and $y$ run from $0$ to $M$.
\end{theorem}
\begin{proof}
Let $\tilde{\mathbb{C}}_{2N}^{\left( j \right)}$ be the codeword set obtained from $\mathbb{C} _{2N}^{\left( j \right)}$ by cyclic left shifting each codeword $j-1$ bits. Since $\mathbb{C} _{2N}^{\left( j \right)}$ is cyclic by Proposition \ref{cyclic_code}, we have $\tilde{\mathbb{C}}_{2N}^{\left( j \right)}=\mathbb{C} _{2N}^{\left( j \right)}$ and this also implies $\tilde{\mathbb{C}}_{2N}^{\left( j \right) \bot}=\mathbb{C} _{2N}^{\left( 2N+2-j \right)}$. For any codeword $\mathbf{\hat{c}}\in \tilde{\mathbb{C}}_{2N}^{\left( j \right)}$, the entries of the Hamming weight constraint $\left( k,l,m,n \right)$ here denote the Hamming weight of $\mathbf{\hat{c}}_{2N-j+2}^{2N-1}$, $\mathbf{\hat{c}}_{2N}$, $\mathbf{\hat{c}}_{1}^{2N-j}$ and $\mathbf{\hat{c}}_{2N-j+1}$, respectively. Applying the generalized MacWilliams identities \cite{simonis1995macwilliams} with respect to a coordinate partition $\mathcal{T} =\left\{ \left\{ 1,\dots ,2N\hspace{-0.1em}-\hspace{-0.1em}j \right\} ,\left\{ 2N\hspace{-0.1em}-\hspace{-0.1em}j\hspace{-0.1em}+\hspace{-0.1em}1 \right\} ,\left\{ 2N\hspace{-0.1em}-\hspace{-0.1em}j\hspace{-0.1em}+\hspace{-0.1em}2,\dots ,2N\hspace{-0.1em}-\hspace{-0.1em}1 \right\} ,\left\{ 2N \right\} \right\}$, the proof is completed.
\end{proof}

Let $T_{2N}^{\left( j \right)}\left( w,r \right)$ designate the number of codewords with information bits weight $w$ and parity bits weight $r$ in $\mathbb{C} _{2N}^{\left( j \right)}$, such that the codeword weight $d=w+r$ and
\begin{equation} \label{Rel_4_split_and_IRWEF}
    T_{2N}^{\left( j \right)}\left( w,r \right) =\sum_{k,l,m,n}{U_{2N}^{\left( j \right)}\left( k,l,m,n \right) \mathbbm{1}_{\varOmega}},
\end{equation}
where $\varOmega =\left\{ \left( k,l,m,n \right) : k+l=w, m+n=r \right\}$.

\begin{proposition}\label{polar_IOWEF_with_IRWEF_and_2split}
For bit-channels with indices $1\leqslant j\leqslant 2N-1$, $A_{2N}^{\left( j \right)}\left( w,d \right)$ satisfies
\begin{equation}
   A_{2N}^{\left( j \right)}\left( w,d \right) =\sum_{p,q}{\left( T_{2N}^{\left( j \right)}\left( w,d-w \right) -V_{2N}^{\left( j+1 \right) \rightarrow \left( j \right)}\left( p,q \right) \right) \mathbbm{1}_{\varOmega}},
\end{equation}
where $\varOmega=\left\{ \left( p,q \right): p=w, q=d-w \right\}$.
\end{proposition}
\begin{proof}
By the definitions of subcode and polar subcode, we have $\mathbb{C} _{2N}^{\left( j \right)}=\mathbb{D} _{2N}^{\left( j \right)}\cup \mathbb{C} _{2N}^{\left( j+1 \right)}$. Then based on the definition of 2-split-spectrum, the proof is straightforward.
\end{proof}

\begin{lemma}\label{4_split_spectrum_Nplus2_2N}
Given the 4-split-spectra $\left\{ U_{N}^{\left( i \right)}\left( k,l,m,n \right) \right\}$ of $\mathbb{C} _{N}^{\left( i \right)}$, $1 \leq i \leq N$. The $\left\{ U_{2N}^{\left( j \right)}\left( e,f,g,h \right) \right\}$ of $\mathbb{C} _{2N}^{\left( j \right)}$ with indices $N+2\leqslant j\leqslant 2N$ can be calculated by $U_{2N}^{\left( j \right)}\left( e,f,g,h \right) =U_{N}^{\left( j-N \right)}\left( k,l,m,n \right)$, where $e=k$, $f=l$, $g=m$ and $h=k+l+m+2n$.
\end{lemma}
\begin{proof}
Based on the Plotkin structure $\left[ \mathbf{u}+\mathbf{v}\left| \mathbf{v} \right. \right]$ of polar codes, any codeword in $\mathbb{C} _{2N}^{\left( j \right)}$ is a repetition of a codeword in $\mathbb{C} _{N}^{\left( j-N \right)}$ when $N+2\leqslant j\leqslant 2N$. In other words, $\forall \mathbf{r}\in \mathbb{C} _{2N}^{\left( j \right)}$, there exists a unique $\mathbf{t}\in \mathbb{C} _{N}^{\left( j-N \right)}$ such that $\mathbf{r}=\left( \mathbf{t},\mathbf{t} \right)$. Then, by the definition of 4-split spectrum, the proof is immediate.
\end{proof}

\begin{lemma}\label{4_split_spectrum_Nplus1}
Since $\mathbb{C} _{N}^{\left( 1 \right)}$ is a linear block code with rate-1, any entry in $\left\{ U_{N}^{\left( 1 \right)}\left( k,l,m,n \right) \right\}$ satisfies $m=n=0$. And we have $U_{2N}^{\left( N+1 \right)}\left( e,f,g,h \right) =U_{N}^{\left( 1 \right)}\left( k,l,0,0 \right)$, where $e=g=k$, $f=h=l$.
\end{lemma}

\begin{lemma}\label{4_split_spectrum_1}
Although $\mathbb{C} _{2N}^{\left( 1 \right)}$ has no dual code, its code rate is 1. Hence, we can obtain its 4-split-spectrum by the symmetry of polar coding, that is, $U_{2N}^{\left( 1 \right)}\left( e,f,0,0 \right) = \binom{2N-1}{f}$.
\end{lemma}

\vspace{-1.0em}
\begin{algorithm}\label{algorithm_4_split}
\caption{Recursive calculation of polar IOWEFs and 4-split-spectra of bit-channels of systematic polar codes}
\KwIn {The 4-split-spectra $\left \{ U_N^{\left(i\right)}\left(k,l,m,n\right) \right \}$ with codelength of $N$, $1\leqslant i\leqslant N$}
\KwOut {The polar IOWEFs $\left \{ A_{2N}^{\left(j\right)}\left(w,d\right) \right \}$ and the 4-split-spectra $\left \{ U_{2N}^{\left(j\right)}\left(e,f,g,h\right) \right \}$ with codelength of $2N$, $1\leqslant j\leqslant 2N$}

Initialize all elements in $\left \{ A_{2N}^{\left(j\right)}\left(w,d\right) \right \}$ and $\left \{ U_{2N}^{\left(j\right)}\left(e,f,g,h\right) \right \}$ to 0, $1\leqslant j\leqslant 2N$\;

\For{$j=2N \to N+2$}
{
\For{each $U_{N}^{\left( j-N \right)}\left( k,l,m,n \right) \in \left\{ \;U_{N}^{\left( j-N \right)}\left( k,l,m,n \right) \right\}$}
{
Calculate the $\left \{ U_{2N}^{\left(j\right)}\left(e,f,g,h\right) \right \}$ by \emph{Lemma \ref{4_split_spectrum_Nplus2_2N}}\;
}
}

\For{each $U_{N}^{\left( 1 \right)}\left( k,l,0,0 \right) \in \left\{ \;U_{N}^{\left( 1 \right)}\left( k,l,0,0 \right) \right\}$}
{
Calculate the $\left \{ U_{2N}^{\left(N+1\right)}\left(e,f,g,h\right) \right \}$ by \emph{Lemma \ref{4_split_spectrum_Nplus1}}\;
}

\For{$j=N \to 2$}
{
Based on the 4-split-spectrum obtained from previous steps, calculate the $\left \{ U_{2N}^{\left(j\right)}\left(e,f,g,h\right) \right \}$ by solving the generalized MacWilliams Identities in \eqref{generalized_MacWilliams}\;
}
Calculate the $\left \{ U_{2N}^{\left(1\right)}\left(e,f,0,0\right) \right \}$ by \emph{Lemma \ref{4_split_spectrum_1}}\;

\For{$j=1 \to 2N-1$}
{
Calculate $\left\{ T_{2N}^{\left( j \right)}\left( w,r \right) \right\}$ and $\left\{ V_{2N}^{\left( j+1 \right) \rightarrow \left( j \right)}\left( p,q \right) \right\}$ by \eqref{Rel_4_split_and_IRWEF} and \eqref{Rel_4_split_and_2_split},respectively\;
}

\For{$j=1 \to 2N-1$}
{
Calculate $\left\{ A_{2N}^{\left( j \right)}\left( w,d \right) \right\}$ by \emph{Proposition \ref{polar_IOWEF_with_IRWEF_and_2split}}\;
}

Calculate $\left\{ A_{2N}^{\left( 2N \right)}\left( w,d \right) \right\}$ by $\left\{ A_{2N}^{\left( 2N \right)}\left( w,d \right) \right\} =\left\{ A_{2N}^{\left( 2N \right)}\left( 1,2N \right) =1 \right\}$\;
		
\textbf{return} The polar IOWEF $\left \{ A_{2N}^{\left(j\right)}\left(w,d\right) \right \}$ and the 4-split-spectrum $\left \{ U_{2N}^{\left(j\right)}\left(e,f,g,h\right) \right \}$
\end{algorithm}
\vspace{-1.0em}

By the definitions of 4-split spectrum, the proofs of Lemma \ref{4_split_spectrum_Nplus1} and Lemma \ref{4_split_spectrum_1} are straightforward and hence omitted.

Algorithm \ref{algorithm_4_split} illustrates the recursive algorithm for calculating the polar IOWEF and 4-split-spectrum corresponding to bit-channel of systematic polar codes. For bit-channel indices $N+1\leqslant j\leqslant 2N$, the 4-split-spectrum of $\mathbb{C} _{2N}^{\left( j \right)}$ can be calculate by Lemma \ref{4_split_spectrum_Nplus2_2N} and Lemma \ref{4_split_spectrum_Nplus1} with the known 4-split spectra of codelength $N$. Followed by these, the generalized MacWilliams identities \eqref{generalized_MacWilliams} can be solved to obtain the 4-split spectrum of $\mathbb{C} _{2N}^{\left( j \right)}$ with $2\leqslant j\leqslant N$. Using Lemma \ref{4_split_spectrum_1}, we have the 4-split-spectrum of $\mathbb{C} _{2N}^{\left( 1 \right)}$. By now, we obtain the complete 4-split-spectra of codelength $2N$. The rest steps aim at calculating the polar IOWEF by Proposition \ref{polar_IOWEF_with_IRWEF_and_2split}. Since $\mathbb{D} _{2N}^{\left( 1 \right)}$ contains only one codeword with Hamming weight $2N$, it is apparent that $\left\{ A_{2N}^{\left( 2N \right)}\left( w,d \right) \right\} =\left\{ A_{2N}^{\left( 2N \right)}\left( 1,2N \right) =1 \right\}$.

The polar IOWEF $\left\{ A_{N}^{\left( i \right)}\left( w,d \right) \right\}$ with the target codelength can be recursively calculated from $N=8$ or $N=16$ of which the polar IOWEF and 4-split-spectrum can be enumerated easily. Benefit from the symmetry and Plotkin structure, the 4-split-spectrum has the following properties which can reduce the calculation complexity. The proofs are immediate and hence omitted.

\begin{proposition}
The number of codewords with odd output weight in $\mathbb{C} _{N}^{\left( i \right)}$ is zero, where $2\leqslant i\leqslant N$. This indicates $U_{N}^{\left( i \right)}\left( k,l,m,n \right) =0$ if $\left( k+l+m+n \right) \%2=0$.
\end{proposition}

\begin{proposition}
Since $\forall \mathbf{c} \in \mathbb{C} _{N}^{\left( i \right)}$, $\mathbf{c}\oplus \left( 1,1,\cdots ,1 \right)$ also belongs to $\mathbb{C} _{N}^{\left( i \right)}$, the 4-split-spectrum is symmetric such that $U_{N}^{\left( i \right)}\left( k,l,m,n \right) =U_{N}^{\left( i \right)}\left( 1-k,N-i-l,1-m,i-2-n \right)$.
\end{proposition}

\section{Bit Error Probability Analysis of Non-systematic Polar Codes}\label{BER_NSPC}

In this section, we move to the analysis of non-systematic polar codes. The IOWEFs of the subcode and polar subcode for non-systematic coding are first introduced. Following this, an upper bound on the conditional bit error probability of bit-channel is derived. Since the accurate calculation of IOWEF is a formidable problem, we also propose an approximation method which is sufficient to calculate the proposed bounds.

\subsection{Conditional Bit Error Probability of Bit Channel for Non-systematic Polar Codes}

Different from the definitions of IOWEFs of the subcode and polar subcode with systematic polar coding, since the information bits are carried by the source word for non-systematic polar codes, the input weight of codeword in $\mathbb{C} _{N}^{\left( i \right)}$ equals the Hamming weight of the information part of source word. In other words, given a codeword $\mathbf{c}=\left( 0_{1}^{i-1},u_{i}^{N} \right) \mathbf{G}_N\in \mathbb{C} _{N}^{\left( i \right)}$, its input weight is $w=wt\left( u_{i}^{N} \right)$ for non-systematic coding. The polar IOWEF of bit-channel holds a similar definition with the input weight of a codeword in $\mathbb{D} _{N}^{\left( i \right)}$ equals to $w=1+wt\left( u_{i+1}^{N} \right)$.

Let $\mathcal{A} _i\subset \mathcal{A}$ denote the subset formed by the elements of $\mathcal{A}$ with values greater than or equal to $i$, and its cardinality $K_i=\left| \mathcal{A} _i \right|\leqslant K$. Furthermore, let $\hat{\mathbf{u}}_{\mathcal{A} _i}$ be the subvector formed by the elements of the wrong decoded information bits $\hat{\mathbf{u}}_{\mathcal{A}}$ with indices in $\mathcal{A} _i$. For each $\mathcal{B} _i$, since the first decoding error occurs at the $i$-th bit-channel and $\hat{u}_{1}^{i-1}=u_{1}^{i-1}$, the corresponding bit error probability is then denoted by $\frac{1}{K}E\left[ d_H(\mathbf{u}_{\mathcal{A} _i},\mathbf{\hat{u}}_{\mathcal{A} _i})|\mathcal{B} _i \right] P_e\left( \mathcal{B} _i \right)$. Following a similar approach in Theorem \ref{BER_SC_polar}, the conditional bit error probability of bit channel for non-systematic polar codes under SC decoding can be represented as
\begin{equation}\label{BER_bit-channel_NSPC}
    P_{b,nsys}^{\mathcal{A}}\left( W_{N}^{\left( i \right)} \right) =\frac{1}{K}E\left[ d_H(\mathbf{u}_{\mathcal{A} _i},\mathbf{\hat{u}}_{\mathcal{A} _i})|\mathcal{E} _i \right] P_e\left( W_{N}^{\left( i \right)} \right).
\end{equation}

However, $d_H(\mathbf{u}_{\mathcal{A} _i},\mathbf{\hat{u}}_{\mathcal{A} _i})$ depends on the specific information set $\mathcal{A}$ which is hence not universal, and the accurate calculation of $d_H(\mathbf{u}_{\mathcal{A} _i},\mathbf{\hat{u}}_{\mathcal{A} _i})$ is a hard and tedious job for even a sufficiently small code dimension $K$. Fortunately, by scaling up the coefficient $\frac{1}{K}E\left[ d_H(\mathbf{u}_{\mathcal{A} _i},\mathbf{\hat{u}}_{\mathcal{A} _i})|\mathcal{E} _i \right]$ and exploiting the concept of polar subcode, we come to the following Proposition.

\begin{proposition}\label{Prop_BER_Bound_BitChannel_NSPC}
For symmetric channels, given the dimension $K$, the conditional bit error probability $P_{b,nsys}\left( W_{N}^{\left( i \right)} \right)$ of a non-systematic polar codes can be bounded by the following inequalities
\begin{equation}\label{BER_bit-channel_IOWEF_NSPC}
    P_{b,nsys}\left( W_{N}^{\left( i \right)} \right) \leqslant \sum_d{\sum_w{\frac{w}{K}A_{N}^{\left( i \right)}\left( w,d \right) P_{N}^{\left( i \right)}\left( d \right)}},
\end{equation}
\begin{equation}\label{appro_BER_bit-channel_IOWEF_NSPC}
    P_{b,nsys}\left( W_{N}^{\left( i \right)} \right) \lessapprox \sum_d{\sum_w{\frac{w}{N-i+1}A_{N}^{\left( i \right)}\left( w,d \right) P_{N}^{\left( i \right)}\left( d \right)}}.
\end{equation}
\end{proposition}
\begin{proof}
Considering $\left\{ \mathcal{E} _i,u_{1}^{N}=0_{1}^{N} \right\} \subset \bigcup_{\mathbf{c}^{\left( 1 \right)}}{\left\{ \tilde{\mathcal{E}}_{i,\mathbf{c}^{\left( 1 \right)}},u_{1}^{N}=0_{1}^{N} \right\}}$, for each bit-channel, we can bound \eqref{BER_bit-channel_NSPC} as
\begin{equation}
    \begin{aligned}
     P_{b,nsys}^{\mathcal{A}}\left( W_{N}^{\left( i \right)} \right)
     & \overset{\left( \mathrm{a} \right)}{=}\frac{1}{K}E\left[ d_H(0_{1}^{K_i},\mathbf{\hat{u}}_{\mathcal{A} _i})|\left( \mathcal{E} _i,u_{1}^{N}=0_{1}^{N} \right) \right] P\left( \mathcal{E} _i|u_{1}^{N}=0_{1}^{N} \right)\\
     & \overset{\left( \mathrm{b} \right)}{\leqslant} \frac{1}{K}\hspace{-0.2em}\sum_{\mathbf{c}^{\left( 1 \right)}}{\hspace{-0.5em}}\sum_{wt(\mathbf{\hat{u}}_{\mathcal{A} _i})}{\hspace{-0.6em}}wt(\mathbf{\hat{u}}_{\mathcal{A} _i})P\hspace{-0.2em}\left( wt(\mathbf{\hat{u}}_{\mathcal{A} _i})|\hspace{-0.2em}\left\{ \tilde{\mathcal{E}}_{i,\mathbf{c}^{\left( 1 \right)}},u_{1}^{N}=0_{1}^{N} \right\} \hspace{-0.2em} \right){\hspace{-0.1em}} P\hspace{-0.2em}\left( \tilde{\mathcal{E}}_{i,\mathbf{c}^{\left( 1 \right)}}|u_{1}^{N}=0_{1}^{N} \right) \\
     & \overset{\left( \mathrm{c} \right)}{=}\sum_{\mathbf{c}^{\left( 1 \right)}}{\frac{wt(\mathbf{u}_{\mathcal{A} _i}^{\left( 1 \right)})}{K}}P_{N}^{\left( i \right)}\left( d \right),
    \end{aligned}
\end{equation}
where $\mathbf{c}^{\left( 1 \right)}=\left( 0_{1}^{i-1},1,u_{i+1}^{N} \right) \mathbf{G}_N=\mathbf{u}^{\left( 1 \right)}\mathbf{G}_N$, $\mathbf{u}_{\mathcal{A} _i}^{\left( 1 \right)}$ is the subvector formed by the elements of $\mathbf{u}^{\left( 1 \right)}$ with indices in $\mathcal{A} _i$. Given any event $\left\{ \tilde{\mathcal{E}}_{i,\mathbf{c}^{\left( 1 \right)}},u_{1}^{N}=0_{1}^{N} \right\}$, $wt( \mathbf{\hat{u}}_{\mathcal{A} _i} )$ is determined and equals to $wt(\mathbf{u}_{\mathcal{A} _i}^{\left( 1 \right)})$, which results in (c) from (b). Since $wt( \mathbf{u}_{\mathcal{A} _i}^{\left( 1 \right)} ) \leqslant wt( {\mathbf{u}^{\left( 1 \right)}}_{i}^{N} ) =w$, we can bound (c) as
\begin{equation}\label{BER_bit-channel_IOWEF_NSPC_tmp}
    P_{b,nsys}\left( W_{N}^{\left( i \right)} \right) \leqslant \sum_{\mathbf{c}^{\left( 1 \right)}}{\frac{w}{K}P_{N}^{\left( i \right)}\left( d \right)}.
\end{equation}
Alternatively, by $K_i\leqslant K$, we can also approximate bound (c) as
\begin{equation}\label{appro_BER_bit-channel_IOWEF_NSPC_tmp}
    P_{b,nsys}\left( W_{N}^{\left( i \right)} \right) \leqslant \sum_{\mathbf{c}^{\left( 1 \right)}}{\frac{wt(\mathbf{u}_{\mathcal{A} _i}^{\left( 1 \right)})}{K_i}P_{N}^{\left( i \right)}\left( d \right)}\lessapprox \sum_{\mathbf{c}^{\left( 1 \right)}}{\frac{w}{N-i+1}P_{N}^{\left( i \right)}\left( d \right)}.
\end{equation}

Combining the codewords with the same weight distribution, we complete the proof.
\end{proof}

From Proposition \ref{Prop_BER_Bound_BitChannel_NSPC}, we can see that for non-systematic polar codes, the conditional bit error probability of bit channel of can also be bounded with the aid of polar IOWEF just as systematic polar codes. In the following subsection, we focus on the calculation of polar IOWEFs of bit-channels for non-systematic polar codes.

\subsection{Calculation of IOWEF and Polar IOWEF of Bit-channel for Non-systematic Polar Codes}

Based on the symmetry and Plotkin structure of polar codes, the IOWEFs $\left\{ S_{2N}^{\left( j \right)}\left( w,d \right) \right\}$ and polar IOWEFs $\left\{ A_{2N}^{\left( j \right)}\left( w,d \right) \right\}$ of systematic polar codes with codelength $2N$ can be recursively calculated by $\left\{ S_{N}^{\left( i \right)}\left( w_2,d_2 \right) \right\}$ and $\left\{ A_{N}^{\left( i \right)}\left( w_1,d_1 \right) \right\}$ as well.

\begin{proposition} \label{polar_IOWEF_Nplus1_2N_NSPC}
Given the polar IOWEFs $\left\{ A_{N}^{\left( i \right)}\left( w_1,d_1 \right) \right\}$ of $\mathbb{D} _{N}^{\left( i \right)}$, $1 \leq i \leq N$. The $\left\{ A_{2N}^{\left( j \right)}\left( w,d \right) \right\}$ with indices $N+1\leqslant j\leqslant 2N$ can be calculated by
\begin{equation}
    A_{2N}^{\left( j \right)}\left( w,d \right) =A_{N}^{\left( j-N \right)}\left( w,\frac{d}{2} \right),
\end{equation}
\end{proposition}
\begin{proof}
For $N+1\leqslant j\leqslant 2N$, based on the Plotkin structure, any codeword $\mathbf{r}$ in $\mathbb{D} _{2N}^{\left( j \right)}$ is a repetition of a codeword $\mathbf{t}$ in $\mathbb{D} _{N}^{\left( j-N \right)}$, i.e., $\mathbf{r}=\left( \mathbf{t},\mathbf{t} \right)$. Given the input weight $w_1$, the number of codewords in $\mathbb{D} _{N}^{\left( j-N \right)}$ with output weight $d_1$ is $A_{N}^{\left( j-N \right)}\left( w_1,d_1 \right)$ which also equals to the number of codewords in $\mathbb{D} _{2N}^{\left( j \right)}$ with output weight $d=2d_1$. The proof is hence completed.
\end{proof}

Also based on the Plotkin structure of polar coding, for bit-channels with indices $1\leqslant j\leqslant N$, any codeword $\mathbf{c}'$ in $\mathbb{D} _{2N}^{\left( j \right)}$ can be represented by $\mathbf{c}'=\left( \mathbf{c}'_1+\mathbf{c}_2\left| \mathbf{c}_2 \right. \right)$, where $\mathbf{c}'_1\in \mathbb{D} _{N}^{\left( j \right)}$ and $\mathbf{c}_2\in \mathbb{C} _{N}^{\left( 1 \right)}$. Although with the knowledge of $\left\{ A_{N}^{\left( j \right)}\left( w_1,d_1 \right) \right\}$ and $\left\{ S_{N}^{\left( 1 \right)}\left( w_2,d_2 \right) \right\}$, the accurate calculation of $\left\{ A_{2N}^{\left( j \right)}\left( w,d \right) \right\}$ of $\mathbb{D} _{2N}^{\left( j \right)}$ still remains a open problem. Inspired by the probabilistic weight distribution \cite{valipour2013probabilistic}, an alternative way is to assume the positions of 1s in $\mathbf{c}'_1$ and $\mathbf{c}_2$ are independent and uniformly distributed. Hence, we resort to the following approximations.

\begin{proposition} \label{polar_IOWEF_1_N_NSPC}
Given the IOWEF $\left\{ S_{N}^{\left( 1 \right)}\left( w_2,d_2 \right) \right\}$ of $\mathbb{C} _{N}^{\left( 1 \right)}$ and polar IOWEFs $\left\{ A_{N}^{\left( i \right)}\left( w_1,d_1 \right) \right\}$ of $\mathbb{D} _{N}^{\left( i \right)}$, $1 \leq i \leq N$. The $\left\{ A_{2N}^{\left( j \right)}\left( w,d \right) \right\}$ with indices $1\leqslant j\leqslant N$ can be approximated by
\begin{equation}
A_{2N}^{\left( j \right)}\left( w,d \right) \approx \sum_{w_1,w_2,d_1,d_2}{\left[ A_{N}^{\left( j \right)}\left( w_1,d_1 \right) S_{N}^{\left( 1 \right)}\left( w_2,d_2 \right) \cdot \sum_{t=\max \left( 0,d_1+d_2-N \right)}^{\min \left( d_1,d_2 \right)}{\frac{\binom{d_2}{t}\binom{N-d_2}{d_1-t}}{\binom{N}{d_1}}} \cdot \mathbbm{1}_{\varOmega} \right]},
\end{equation}
where $\varOmega =\left\{ \left( w_1,w_2,d_1,d_2 \right) :w=w_1+w_2,d=d_1+2d_2 \right\}$.
\end{proposition}
\begin{proof}
With the independency assumption, the Hamming weight $d_1$ of $\mathbf{c}'_1\in \mathbb{D} _{N}^{\left( j \right)}$ yields $\binom{N}{d_1}$ possible codewords with equal probability ${1}/{\binom{N}{d_1}}$. Among these codewords, let $t$ be the number of positions at which the elements in $\mathbf{c}'_1$ and $\mathbf{c} _2$ are both equal to $1$, where $\mathbf{c}_2\in \mathbb{C} _{N}^{\left( 1 \right)}$. The minimum and maximum values of $t$ are $\max \left( 0,d_1+d_2-N \right)$ and $\min \left( d_1,d_2 \right)$, respectively. Given the value $t$, it is easy to check that the Hamming weigh of $\mathbf{c}'=\left( \mathbf{c}'_1+\mathbf{c}_2\left| \mathbf{c}_2 \right. \right)$ is $d_1+2d_2-t$, and this yields total $\binom{d_2}{t}\binom{N-d_2}{d_1-t}$ combinations of $\mathbf{c}'_1$ and $\mathbf{c} _2$. Hence, the probability of $\mathbf{c}'$ with Hamming weight $d_1+2d_2-t$ is ${\binom{d_2}{t}\binom{N-d_2}{d_1-t}}/{\binom{N}{d_1}}$. Finally, the combination of $\mathbf{c}'_1$ and $\mathbf{c} _2$ gives the input weight of $\mathbf{c}'$ as $w_1+w_2$. The proof is completed.
\end{proof}

\begin{proposition} \label{IOWEF_NSPC}
For $1\leqslant j\leqslant 2N-1$, the IOWEF $\left\{ S_{2N}^{\left( j \right)}\left( w,d \right) \right\}$ can be calculated by $S_{2N}^{\left( j \right)}\left( w,d \right) \hspace{-0.2em} =S_{2N}^{\left( j+1 \right)}\left( w,d \right) +A_{2N}^{\left( j \right)}\left( w,d \right)$. Particularly, $\left\{ S_{2N}^{\left( 2N \right)}\left( w,d \right) \right\} =\left\{ S_{2N}^{\left( 2N \right)}\left( 0,0 \right) =1,S_{2N}^{\left( 2N \right)}\left( 1,2N \right) =1 \right\}$.
\end{proposition}
\begin{proof}
The proof is based on the fact that $\mathbb{C} _{2N}^{\left( j \right)}=\mathbb{C} _{2N}^{\left( j+1 \right)}\cup \mathbb{D} _{2N}^{\left( j \right)}$ and is straightforward.
\end{proof}

Based on the properties of polar coding, the approximated polar IOWEFs can be further revised by the following lemma, which is proved in Appendix \ref{proof_of_polar_IOWEF_revise_NSPC}.

\begin{lemma} \label{polar_IOWEF_revise_NSPC}
For each bit-channel $i$, the maximum input weight $w_{\max}=N-i+1$ of $\mathbb{D} _{N}^{\left( i \right)}$ results in the codeword with output weight $d=wt\left( \mathbf{g}_{i,N} \right)$, where $\mathbf{g}_{i,N}$ is the $i$-th row of $\mathbf{G}_N$. That is,
\begin{equation}
    A_{N}^{\left( i \right)}\left( N-i+1,d \right) =
    \begin{cases}
      1, & \text{if } d=wt\left( \mathbf{g}_{i,N} \right)\\
      0, & \mbox{otherwise}.
    \end{cases}.
\end{equation}
\end{lemma}

\begin{remark}
Note that one may also investigate other revise regulations to make $\left\{ A_{2N}^{\left( j \right)}\left( w,d \right) \right\}$ approximate its true values. However, as we will see, the above approximations are sufficient to compute the bit error probability of non-systematic polar codes.
\end{remark}

\begin{algorithm}[h] \label{algorithm_polar_IOWEF_NSPC}
\caption{Recursive calculation of IOWEFs and polar IOWEFs of bit-channels of non-systematic polar codes}
\KwIn {The polar IOWEFs $\left\{ A_{N}^{\left( i \right)}\left( w_1,d_1 \right) \right\}$ and IOWEFs $\left\{ S_{N}^{\left( i \right)}\left( w_2,d_2 \right) \right\}$ of $N$, $1\leqslant i\leqslant N$}
\KwOut {The polar IOWEFs $\left\{ A_{2N}^{\left( j \right)}\left( w,d \right) \right\}$ and IOWEFs $\left\{ S_{2N}^{\left( j \right)}\left( w,d \right) \right\}$ of $2N$, $1\leqslant j\leqslant 2N$}

Initialize all elements in $\left \{ A_{2N}^{\left(j\right)}\left(w,d\right) \right \}$ and $\left\{ S_{2N}^{\left( j \right)}\left( w,d \right) \right\}$ to 0, $1\leqslant j\leqslant 2N$\;

\For{$j=2N \to N+1$}
{
\For{each $A_{N}^{\left( j-N \right)}\left( w_1,d_1 \right) \in \left\{ A_{N}^{\left( j-N \right)}\left( w_1,d_1 \right) \right\}$}
{
Calculate the $\left\{ A_{2N}^{\left( j \right)}\left( w,d \right) \right\}$ by \emph{Proposition \ref{polar_IOWEF_Nplus1_2N_NSPC}}\;
}
}

\For{$j=N \to 1$}
{
Calculate the approximation of $\left\{ A_{2N}^{\left( j \right)}\left( w,d \right) \right\}$ by \emph{Proposition \ref{polar_IOWEF_1_N_NSPC}}\;

Revise $A_{2N}^{\left( j \right)}\left( 2N-j+1,d \right)$ by \emph{Lemma \ref{polar_IOWEF_revise_NSPC}}\;
}

Set $\left\{ S_{2N}^{\left( 2N \right)}\left( w,d \right) \right\} =\left\{ S_{2N}^{\left( 2N \right)}\left( 0,0 \right) =1,S_{2N}^{\left( 2N \right)}\left( 1,2N \right) =1 \right\}$.

\For{$j=2N-1 \to 1$}
{
Calculate the $\left\{ S_{2N}^{\left( j \right)}\left( w,d \right) \right\}$ by \emph{Proposition \ref{IOWEF_NSPC}}\;
}

\textbf{return} The polar IOWEFs $\left\{ A_{2N}^{\left( j \right)}\left( w,d \right) \right\}$ and IOWEFs $\left\{ S_{2N}^{\left( j \right)}\left( w,d \right) \right\}$
\end{algorithm}

The recursive calculation of the approximate IOWEFs and polar IOWEFs of bit-channels are summarized as Algorithm \ref{algorithm_polar_IOWEF_NSPC}. For bit-channel indices $N+1\leqslant j\leqslant 2N$, the polar IOWEF of $\mathbb{D} _{2N}^{\left( j \right)}$ is calculated by Proposition \ref{polar_IOWEF_1_N_NSPC} with the known polar IOWEFs of codelength $N$. Along with the $\left\{ S_{N}^{\left( 1 \right)}\left( w_2,d_2 \right) \right\}$, the polar IOWEFs of bit-channels with indices $1\leqslant j\leqslant N$ can then be approximately calculated by Proposition \ref{polar_IOWEF_1_N_NSPC}, and Lemma \ref{polar_IOWEF_revise_NSPC} can be further adopted to revise the values. Since $\mathbb{C} _{2N}^{\left( j \right)}=\mathbb{C} _{2N}^{\left( j+1 \right)}\cup \mathbb{D} _{2N}^{\left( j \right)}$ and by Proposition \ref{IOWEF_NSPC}, we have the  IOWEFs of $\mathbb{C} _{2N}^{\left( j \right)}$. If the target codelength is smaller than $32$, the polar IOWEFs and IOWEFs can be enumerated easily. While for the other cases, they can then be recursively calculated by the above algorithm.

\section{Construction Methods Based on the Conditional Bit Error Probability}

In this section, based on the upper bound on the conditional bit error probability, we proposed two types of construction metrics which have explicit forms and linear computational complexity.

\subsection{Construction Metrics based on Union-Bhattacharyya Bound}

The conditional bit error probability upper bounds of bit-channels can be regarded as reliability metrics to determine the information sets for data transmission. For the binary-input AWGN (BI-AWGN) channels , the pairwise error probability with Hamming weight $d$ is equal to
\begin{equation}
    P_N^{\left( i \right)}\left( d \right) = Q\left( {\sqrt {\frac{{2d{E_s}}}{{{N_0}}}} } \right),
\end{equation}
where $\frac{E_s}{N_0}$ is the symbol signal-to-noise ratio (SNR) and $Q\left( x \right) = \frac{1}{{\sqrt {2\pi } }}\int_x^\infty  {{e^{ - \frac{{{t^2}}}{2}}}} dt$ is the probability that a random Gaussian variable with zero mean and unit variance exceeds the value $x$. Exploiting the Chernoff bound on the $Q$-function in \eqref{Bit_Error_Probability_BitChannel_d_form1} and \eqref{Bit_Error_Probability_BitChannel_d_form2} (i.e., $Q\left( x \right) \le \exp \left( {{{ - {x^2}} \mathord{\left/{\vphantom {{ - {x^2}} 2}} \right. \kern-\nulldelimiterspace} 2}} \right), x>0$) yields the union-Bhattacharyya bounds that are expressed as
\begin{align}
    P_{b,sys}\left( W_{N}^{\left( i \right)} \right)
    & \leqslant \sum_d{\sum_w{\frac{w}{N-i+1}A_{N}^{\left( i \right)}\left( w,d \right) \exp \left( -\frac{dE_s}{N_0} \right)}} \label{Bit_Error_Probability_BitChannel_UB_IOWEF}\\
    & =\sum_d{\frac{d}{N}A_{N}^{\left( i \right)}\left( d \right) \exp \left( -\frac{dE_s}{N_0} \right)} \label{Bit_Error_Probability_BitChannel_UB_WEF}
\end{align}

Compared with \eqref{Bit_Error_Probability_BitChannel_d_form1} and \eqref{Bit_Error_Probability_BitChannel_d_form2}, the above bounds provide a slightly looser but simpler form to analyze the performance of bit-channels. Note that the approximation of the $Q$-function with an exponential is not sufficiently tight, the union-Bhattacharyya bounds will not be further considered for evaluating BER performance in the following section. To further facilitate the practical implementation, by taking the logarithmic form of the union-Bhattacharyya bound \eqref{Bit_Error_Probability_BitChannel_UB_WEF} and using the max-log approximation of Jacobian logarithm, we have
\begin{equation}
    \ln \left\{ \sum_d{\frac{d}{N}A_{N}^{\left( i \right)}\left( d \right) \exp \left( -\frac{dE_s}{N_0} \right)} \right\} \approx \underset{d}{\max}\left\{ \ln \frac{d}{N}+\ln A_{N}^{\left( i \right)}\left( d \right) -\frac{dE_s}{N_0} \right\}.
\end{equation}

Based on these, we propose a new construction metric named the union-Bhattacharyya bound weight of the bit error probability (UBWB) as follows.

 \noindent \textbf{Metric 1:} The UBWB of the $i$-th bit-channel for systematic polar codes is defined as
\begin{equation}
    \mathrm{UBWB}_{N,sys}^{\left( i \right)} = \underset{d}{\max}\left\{ \ln \frac{d}{N}+\ln A_{N}^{\left( i \right)}\left( d \right) -\frac{dE_s}{N_0} \right\},
\end{equation}
where $i=1,2,\ldots ,N$, $\left\{ A_{N}^{\left( i \right)}\left( d \right) \right\}$ is the polar spectrum defined in \eqref{def_polar_WEF}.

\begin{remark}
Applying a similar approach to \eqref{Bit_Error_Probability_BitChannel_UB_IOWEF}, one can also obtain a similar metric that are denoted as $\mathrm{UBWB}_{N,sys}^{\left( i \right)} = \underset{w,d}{\max}\left\{ \ln \frac{w}{N-i+1}+\ln A_{N}^{\left( i \right)}\left( w,d \right) -\frac{dE_s}{N_0} \right\}$. However, since the polar spectrum is much easier to calculate than the polar IOWEF, the UBWB in the form of $A_{N}^{\left( i \right)}\left( d \right)$ is more preferred in practical implementation.
\end{remark}

Similarly, for non-systematic polar codes, the union-Bhattacharyya bound and approximate union-Bhattacharyya bound on the conditional bit error probability are denoted respectively by
\begin{equation}\label{BER_bit-channel_UB_IOWEF_NSPC}
    P_{b,nsys}\left( W_{N}^{\left( i \right)} \right) \leqslant \sum_d{\sum_w{\frac{w}{K}A_{N}^{\left( i \right)}\left( w,d \right) \exp \left( -\frac{dE_s}{N_0} \right)}},
\end{equation}
\begin{equation}\label{appro_BER_bit-channel_UB_IOWEF_NSPC}
    P_{b,nsys}\left( W_{N}^{\left( i \right)} \right) \lessapprox \sum_d{\sum_w{\frac{w}{N-i+1}A_{N}^{\left( i \right)}\left( w,d \right)\exp \left( -\frac{dE_s}{N_0} \right)}}.
\end{equation}

Note that \eqref{BER_bit-channel_UB_IOWEF_NSPC} depends on the code dimension $K$ which is not convenient for obtaining a universal reliability sequence. Therefore, based on \eqref{appro_BER_bit-channel_UB_IOWEF_NSPC}, we propose the following UBWB metric for non-systematic polar codes as

 \noindent \textbf{Metric 2:} The UBWB of the $i$-th bit-channel for non-systematic polar codes is defined as
\begin{equation}
    \mathrm{UBWB}_{N,nsys}^{\left( i \right)} = \underset{w,d}{\max}\left\{ \ln \frac{w}{N-i+1}+\ln A_{N}^{\left( i \right)}\left( w,d \right) -\frac{dE_s}{N_0} \right\},
\end{equation}
where $i=1,2,\ldots ,N$, $\left\{A_{N}^{\left( i \right)}\left( w,d \right)\right\}$ is the polar IOWEF for non-systematic polar codes.

For other symmetric B-DMCs, by investigating the pairwise error probability $P_{N}^{\left( i \right)}\left( d \right)$, a similar approach can be directly adopted to obtain the corresponding construction metrics.

\subsection{Construction Metrics based on Simplified Union-Bhattacharyya Bound}

The minimum Hamming weight $d_{\min}$ dominates the union-Bhattacharyya bound in high SNR regime. By fixing the output weight $d=d_{\min}$, one can obtain the simplified union-Bhattacharyya bound and we further propose the simplified UBWB (SUBWB) metrics as follows.

 \noindent \textbf{Metric 3:} The SUBWB of the $i$-th bit-channel for systematic polar codes is defined as
\begin{equation}
    \mathrm{SUBWB}_{N,sys}^{\left( i \right)} = \ln \frac{d_{\min}}{N}+\ln A_{N}^{\left( i \right)}\left( d_{\min} \right) -\frac{d_{\min}E_s}{N_0},
\end{equation}
where $i=1,2,\ldots ,N$, $A_{N}^{\left( i \right)}\left( d_{\min} \right)$ is the number of codewords with output weight $d_{\min}$ in $\mathbb{D} _{N}^{\left( i \right)}$.

 \noindent \textbf{Metric 4:} The SUBWB of the $i$-th bit-channel for non-systematic polar codes is defined as
\begin{equation}
    \mathrm{SUBWB}_{N,nsys}^{\left( i \right)} = \underset{w}{\max}\left\{ \ln \frac{w}{N-i+1}+\ln A_{N}^{\left( i \right)}\left( w,d_{\min} \right) -\frac{d_{\min}E_s}{N_0} \right\},
\end{equation}
where $i=1,2,\ldots ,N$ and $A_{N}^{\left( i \right)}\left( w,d_{\min} \right)$ is the number of codewords in $\mathbb{D} _{N}^{\left( i \right)}$ with input weight $w$ and output weight $d_{min}$ of non-systematic coding.

Compared with DE and GA which involve high-dimensional recursive calculations, the UBWB and SUBWB metrics have explicit analytical forms and linear computational complexity since the polar IOWEF and polar spectrum can be precomputed and stored. Besides, one can also obtain a fixed reliability sequence (like polar sequence in \cite{3gpp.38.212}) by selecting a suitable design-SNR.

\section{Numerical Analysis and Simulation Results}

In this section, we first provide numerical analysis of the bit error probability upper bounds described in Section \ref{BER_SPC} and \ref{BER_NSPC}. Then, simulation results based on various construction methods of polar codes are also compared.

\subsection{Numerical Analysis of Bit Error Probability Upper Bounds}

To examine the effectiveness of the bit error probability upper bounds of both systematic and non-systematic polar codes, we present the corresponding upper bounds under SC decoding along with the relevant simulation results over BI-AWGN channels.
\begin{figure}[htp]
  \centering{\includegraphics[scale=0.53]{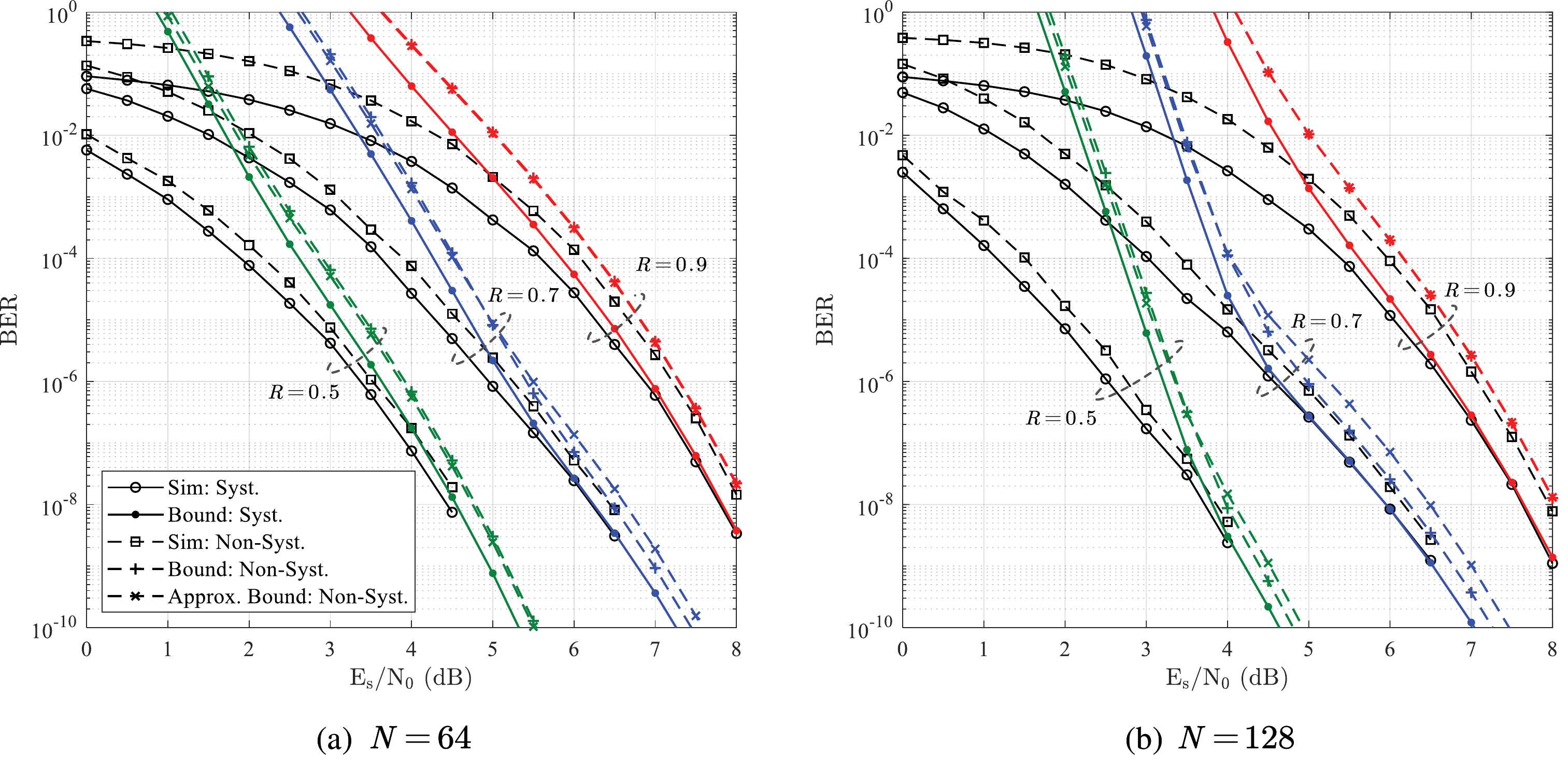}}
  \caption{Bit error probability upper bounds and simulation results of polar codes under SC decoding.}\label{SPC_NSPC_bounds}
\end{figure}

The results for both systematic and non-systematic polar codes are shown in Fig. \ref{SPC_NSPC_bounds}, while considering the codelength of $N=64$ and $N=128$. The information set is determined by the polar sequence in 5G NR with code rate of $R\in \left\{ 0.5,0.7,0.9 \right\}$. Note that here we focus on validating the upper bounds but not on comparing the construction metrics, hence we just use the polar sequence as an example. The upper bounds of systematic and non-systematic polar codes are calculated by \eqref{Bit_Error_Probability_BitChannel_d_form2} and \eqref{BER_bit-channel_IOWEF_NSPC} respectively, and the approximate bounds of non-systematic polar codes are calculated by \eqref{appro_BER_bit-channel_IOWEF_NSPC}. As shown in Fig. \ref{SPC_NSPC_bounds}, systematic polar codes achieve better BER performance than non-systematic polar codes, and the upper bounds can well reflect the tendency of the simulated BER curve at high SNRs. While as the SNR decreases, the upper bounds tend to diverge. This is due to the conceptual weakness of union bounds \cite{sason2006performance} that the intersections of decision regions related to the codewords other than the actual transmitted one are counted more than once. Despite all this, the union bound
serves as a basis for many other bounding techniques. Furthermore, the comparison between the exact bound and the approximate bound of non-systematic polar codes reveals that they are in good agreement, especially for high code rates. This indicates that the approximate upper bound can be an efficient tool to evaluate the reliability order of bit-channels with respect to the bit error probability.

\subsection{Simulation Results}

In this section, the BER simulation performance of polar codes constructed by the proposed UBWB/SUBWB metrics and the conventional methods are compared. The codelength $N$ is set to $256$ with code rate of $R\in \left\{ 1/3,1/2,2/3 \right\}$. The SC and SCL decodings with the list size of $L=32$ are adopted to decode polar codes. For the GA method, polar codes are constructed individually for each evaluated SNR.
\begin{figure}[htb]
  \centering{\includegraphics[scale=0.53]{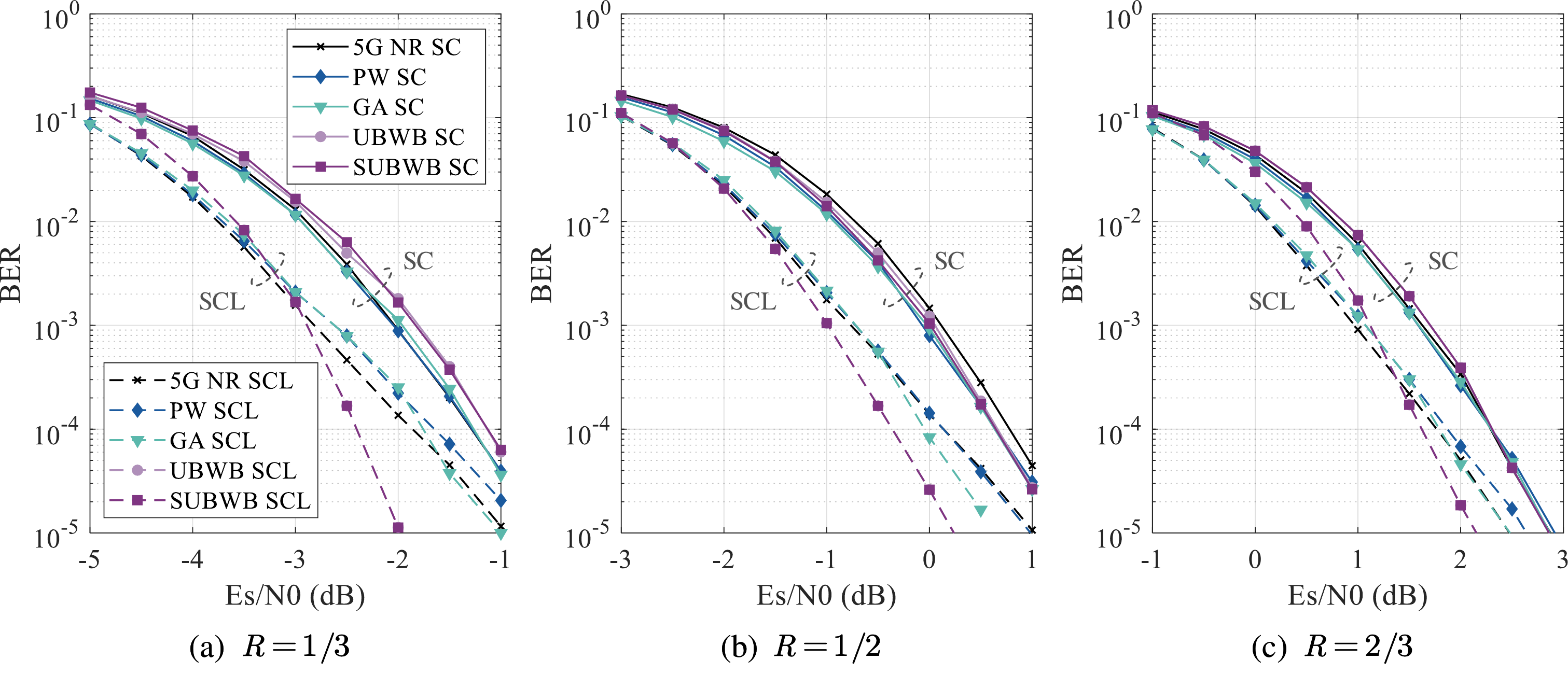}}
  \caption{BER performance comparisons among systematic polar codes constructed by various metrics.}\label{BER_SPC_Comparisons}
\end{figure}

Fig. \ref{BER_SPC_Comparisons} provides a pictorial description of the BER performance of systematic polar codes constructed by various metrics. For SC decoding, the design-SNRs of UBWB/SUBWB for $R=1/3$, $1/2$ and $2/3$ are set to 2.0, 3.5 and 4.5 dB, respectively; while for SCL decoding, are 3.5, 5.5 and 7.0 dB. It is observed that systematic polar codes constructed by UBWB/SUBWB can achieve comparable performance to those constructed by conventional methods under SC decoding, and better performance under SCL decoding. More explicitly, for $R=1/3$ and $\rm{BER}=10^{-4}$, UBWB and SUBWB outperform GA by about 0.64 dB under SCL decoding. This is because UBWB and SUBWB benefit from the adequate utilization of the weight distributions of polar subcodes, which might also result in an improvement in distance spectrum of polar codes. Another kind of comparisons, shown in Fig. \ref{BER_NSPC_Comparisons}, is performed among non-systematic polar codes. The design-SNRs of UBWB/SUBWB for $R=1/3$, $1/2$ and $2/3$ are 1.5, 3.5 and 4.0 dB under SC decoding; 3.5, 5.5 and 7.0 dB under SCL decoding. The results in Fig. \ref{BER_NSPC_Comparisons} are similar to those of systematic polar codes and these also indicate that the approximations of the polar IOWEF and conditional bit error probability upper bounds are sufficient to evaluate the performance of bit-channels.
\begin{figure}[htb]
  \centering{\includegraphics[scale=0.53]{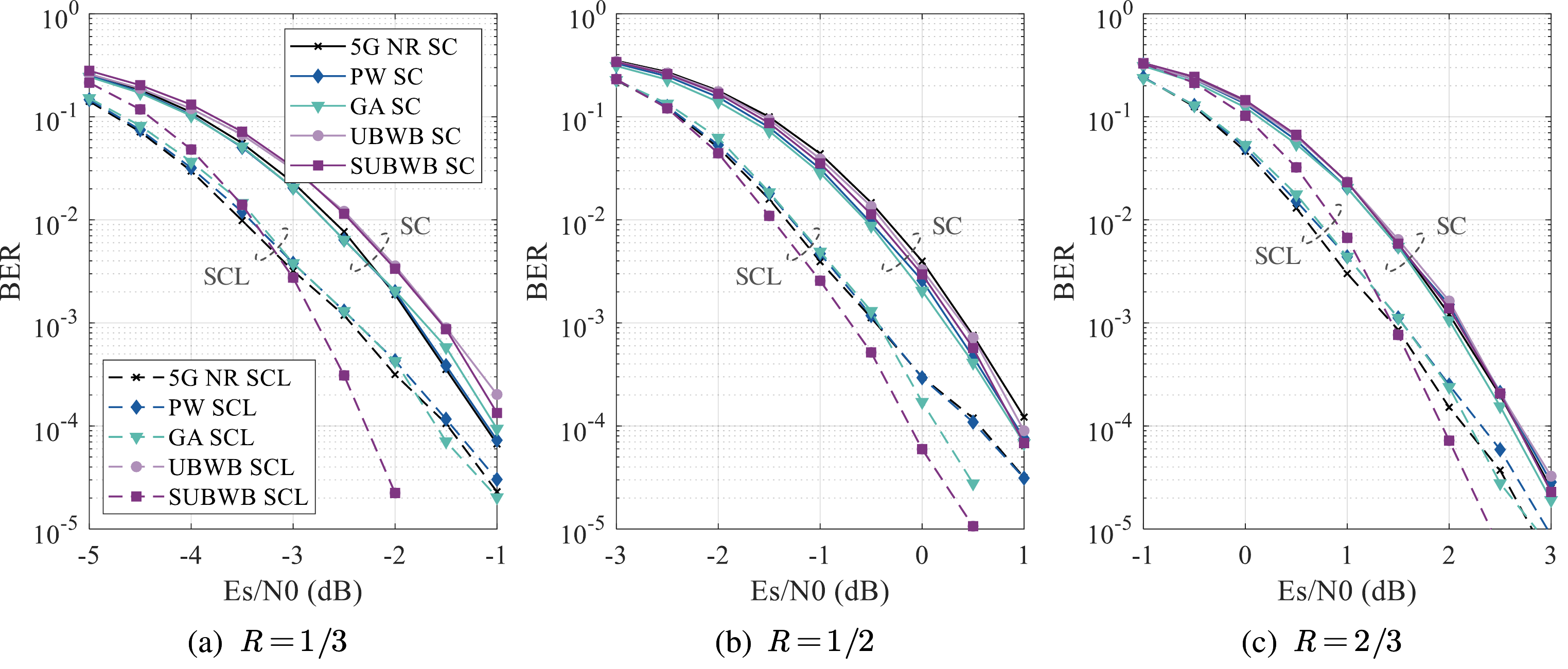}}
  \caption{BER performance comparisons among non-systematic polar codes constructed by various metrics.}\label{BER_NSPC_Comparisons}
\end{figure}

In a nutshell, the proposed UBWB and SUBWB are two reasonable construction metrics that can achieve superior performance than conventional methods under SCL decoding. Moreover, UBWB and SUBWB are more preferred than GA due to the linear computational complexity.

\section{Conclusion}
In this paper, we introduce the conditional bit error probability of bit-channel to analyze the bit error performance of both systematic and non-systematic polar codes. Then, the upper bounds on the conditional bit error probability of bit-channel and the bit error probability of polar codes under SC decoding are also given. Based on these, two construction metrics are proposed, which have linear computational complexity, explicit forms, and satisfying performance. We believe this work would be helpful for analyzing the bit error performance of other polar coded systems.

\appendix

\subsection{Proof of Proposition \ref{cyclic_code}} \label{proof_of_cyclic_code}
We first use mathematical induction to prove the subcode is cyclic. For $N=2$, it can be easily proved by enumerating all possible codewords that $\mathbb C_2^{\left( i \right)}$ is cyclic, where $i=1,2$. Assume the statement is true for $N \ge 2$, that is, $\mathbb C_N^{\left( i \right)}$ is cyclic for any $1 \le i \le N$. Now, we need to prove that the subcode of codelength $2N$ is also a cyclic code.

For subcode $\mathbb C_{2N}^{\left( i \right)}$ with index ${N+1} \le i \le 2N$, based on the Plotkin  structure $\left[ {{\bf{u}} + {\bf{v}}|{\bf{v}}} \right]$ of polar coding, it is easy to prove that any codeword in $\mathbb C_{2N}^{\left( i \right)}$ is a repetition of a unique codeword in $\mathbb C_N^{\left( {i - N} \right)}$. In other words, $\forall {\bf{r}} = \left( {{r_1},{r_2}, \cdots ,{r_{2N}}} \right) \in \mathbb C_{2N}^{\left( i \right)}$ with ${N+1} \le i \le 2N$, it satisfies that ${\bf{r}} = \left( {{\bf{t}},{\bf{t}}} \right)$, where ${\bf{t}} = \left( {{t_1},{t_2}, \cdots ,{t_N}} \right) \in \mathbb C_N^{\left( {i - N} \right)}$, and vice versa. The cyclic shift of $\bf{r}$ is denoted by ${{\bf{r}}^{\left( 1 \right)}} = \left( {{t_N},{t_1}, \cdots ,{t_{N - 1}},{t_N},{t_1}, \cdots ,{t_{N - 1}}} \right)$. Since $\mathbb C_N^{\left( i-N \right)}$ is a cyclic code and ${\bf{t}} \in \mathbb C_N^{\left( {i - N} \right)}$, we have ${{\bf{t}}^{\left( 1 \right)}} = \left( {{t_N},{t_1}, \cdots ,{t_{N - 1}}} \right) \in \mathbb C_N^{\left( {i - N} \right)}$. This implies ${{\bf{r}}^{\left( 1 \right)}} = \left( {{{\bf{t}}^{\left( 1 \right)}},{{\bf{t}}^{\left( 1 \right)}}} \right)$ is still a codeword in $\mathbb C_{2N}^{\left( i \right)}$. Hence we have proved that $\mathbb C_{2N}^{\left( i \right)}$ is cyclic for ${N+1} \le i \le 2N$. Moreover, since the dual code of a cyclic code is also cyclic \cite[Th. 4.2.6]{huffman2010fundamentals} and by Proposition \ref{subcode_duality}, it can be concluded that $\mathbb C_{2N}^{\left( i \right)}$ is cyclic for $2 \le i \le N$. For $\mathbb C_{2N}^{\left( 1 \right)}$, it consists all $2^{2N}$ possible codewords of codelength $2N$, hence $\mathbb C_{2N}^{\left( 1 \right)}$ is cyclic as well. This completes the proof that any subcode is a cyclic code.

To prove the polar subcode is also cyclic, observe that $\mathbb C_N^{\left( i \right)} = \mathbb C_N^{\left( {i + 1} \right)} \cup \mathbb D_N^{\left( i \right)}$ for $1 \le i \le N - 1$. For any codeword ${\bf{s}} \in \mathbb D_N^{\left( i \right)}$ and its cyclic shift ${{\bf{s}}^{\left( 1 \right)}}$, we have ${{\bf{s}}^{\left( 1 \right)}} \in \mathbb C_N^{\left( i \right)}$. Assume ${{\bf{s}}^{\left( 1 \right)}} \in \mathbb C_N^{\left( i+1 \right)}$, since $\mathbb C_N^{\left( i+1 \right)}$ is cyclic, then ${\bf{s}} \in \mathbb C_N^{\left( i+1 \right)}$ which is contradict to the fact that ${\bf{s}} \in \mathbb D_N^{\left( i \right)}$. Hence, ${{\bf{s}}^{\left( 1 \right)}} \in \mathbb D_N^{\left( i+1 \right)}$ and this implies $\mathbb D_N^{\left( i+1 \right)}$ is a cyclic code, where $1 \le i \le N - 1$. For $\mathbb D_N^{\left( 1 \right)}$, it only contains an all-ones codeword and is obviously also cyclic. This proves the second claim that any polar subcode is also cyclic.

\subsection{Proof of Lemma \ref{polar_IOWEF_revise_NSPC}}\label{proof_of_polar_IOWEF_revise_NSPC}
This can also be proved by mathematical induction. First, let $\mathbf{\dot{c}}'_{i,N}=\left( 0_{1}^{i-1},1_{i}^{N} \right) \mathbf{G}_N$ be the codeword generated by the maximum input weight of $\mathbb{D} _{N}^{\left( i \right)}$. Since $\mathbf{G}_N$ is a lower triangular matrix, the last bit in $\mathbf{\dot{c}}'_{i,N}$ is $1$. Particularly, $\mathbf{\dot{c}}'_{1,N}=\left( 0_{1}^{N-1},1 \right)$ since the Hamming weight of each column of $\mathbf{G}_N$ except the last one is even. For the base case. It is easy to check that the Hamming weight of $\mathbf{\dot{c}}'_{i,2}$ equals to $wt\left( \mathbf{g}_{i,2} \right)$, where $i=1,2$. Assume the statement is true for $N \geqslant 2$, we need to prove this still holds for $2N$. For $\mathbb{D} _{2N}^{\left( i \right)}$ with index $N+1\leqslant i\leqslant 2N$, based on the Plotkin structure, we have $\mathbf{\dot{c}}'_{i,2N}=\left( \mathbf{\dot{c}}'_{i-N,N},\mathbf{\dot{c}}'_{i-N,N} \right)$. Since $wt\left( \mathbf{\dot{c}}'_{i-N,N} \right) =wt\left( \mathbf{g}_{i-N,N} \right)$ and $\mathbf{g}_{i,2N}=\left( \mathbf{g}_{i-N,N},\mathbf{g}_{i-N,N} \right)$, we get $wt\left( \mathbf{\dot{c}}'_{i,2N} \right) =wt\left( \mathbf{g}_{i,2N} \right)$. For the case of $1\leqslant i\leqslant N$, we have $\mathbf{\dot{c}}'_{i,2N}=\left( \mathbf{\dot{c}}'_{i,N}+\mathbf{\dot{c}}'_{1,N},\mathbf{\dot{c}}'_{1,N} \right)$. Since the last bit in $\mathbf{\dot{c}}'_{i,N}$ is $1$ and $\mathbf{\dot{c}}'_{1,N}=\left( 0_{1}^{N-1},1 \right)$, the Plotkin structure gives the Hamming weight of $\mathbf{\dot{c}}'_{i,2N}$ as $wt\left( \mathbf{\dot{c}}'_{i,2N} \right) =wt\left( \mathbf{\dot{c}}'_{i,N} \right) =wt\left( \mathbf{g}_{i,N} \right) =wt\left( \mathbf{g}_{i,2N} \right)$. This proves the inductive step and completes the proof.

\ifCLASSOPTIONcaptionsoff
  \newpage
\fi

\bibliographystyle{ieeetr}
\bibliography{reference}

\end{document}